\documentclass{sig-alternate-2013-modified}
\usepackage{url}
\usepackage{amssymb,amsmath,amsfonts}
\usepackage{amsmath}
\usepackage{graphicx}
\usepackage{times}
\usepackage{algorithm}
\usepackage[algo2e, ruled, vlined]{algorithm2e}

\newcommand{\ignore}[1]{}
\newcommand{\notinproc}[1]{}
\newcommand{\onlyinproc}[1]{#1}

\newtheorem{thm}{Theorem}[section]
\newtheorem{theorem}{Theorem}[section]

\newtheorem{lemma}[thm]{Lemma}

\newtheorem{corollary}[thm]{ Corollary}

\newcommand{\ADS}{\mathop{\rm ADS}}

\newcommand{\kth}{\text{k}^{\text{th}}}
\newcommand{\Inf}{\mathop{\rm Inf}}

\newcommand{\rN}{{\sc rR}}

\newcommand{\rD}{{\sc rR}}
\newcommand{\D}{{\sc D}}
\newcommand{\skim}{{\sc SKIM}}
\newcommand{\instance}[1]{\textsf{#1}}

\begin{document}

\title{Reverse Ranking by Graph Structure: \\ Model and Scalable
  Algorithms}

\date{\today}

\numberofauthors{2}
\author{
\alignauthor Eliav Buchnik\\
       \affaddr{School of Computer Science}\\
       \affaddr{Tel Aviv University, Israel}\\
       \email{eliavbuh@gmail.com}
\alignauthor Edith Cohen\\
       \affaddr{Google Research, CA, USA}\\
       \affaddr{Tel Aviv University, Israel}\\
       \email{edith@cohenwang.com}
}

\ignore{
\author{
Eliav Buchnik \thanks{School of Computer Science, Tel Aviv University,
  Israel {\tt eliavbuh@gmail.com, edith@cohenwang.com}}  \and
Edith Cohen$^*$
}}

\maketitle 

\begin{abstract}
{\small 
 Distances in a network capture relations between nodes 
 and are the basis of centrality, similarity, and influence measures.
Often, however, the relevance of a node $u$ to a node
$v$ is more precisely measured not by the magnitude of the distance, but
by the number of nodes that are closer to $v$ than $u$. That is,
by the {\em rank} of $u$ in an ordering of nodes by increasing distance
from $v$.  

We identify and address fundamental challenges in rank-based
graph mining.  We first consider single-source computation of reverse-ranks and
design a ``Dijkstra-like''  algorithm which computes
nodes in order of  increasing approximate reverse rank while only
traversing edges adjacent to returned nodes.
 We then define {\em reverse-rank influence}, which naturally extends reverse 
nearest neighbors influence [Korn and Muthukrishnan 2000] and builds 
on a well studied distance-based influence. We present 
near-linear algorithms for 
greedy approximate reverse-rank influence maximization. The design relies on 
our single-source algorithm.
 Our algorithms utilize near-linear 
preprocessing of the network to compute all-distance sketches.
As a contribution of independent interest, we present  a
novel algorithm for computing these sketches, which have many other
applications,  on multi-core
architectures. 

 We complement our algorithms by establishing the hardness of 
 computing {\em exact} reverse-ranks for a single source and {\em
   exact}
 reverse-rank influence. This implies that when
using near-linear algorithms,  the small relative errors we obtain are 
the best we can currently hope for.

 Finally, we conduct an experimental evaluation on graphs with tens of
 millions of edges, demonstrating both scalability and accuracy.
}
   \end{abstract}

\section{Introduction}
\SetKwData{rNrNN}{NN$_{\text{\rN}}$}
\SetKw{Or}{or}
\SetKw{And}{and}

Shortest-paths distances in a network
are a classic measure of the relation between nodes and
are the basis of
similarity \cite{Liben-Nowell_Kleinberg:CIKM2003,CDFGGW:COSN2013}, centrality
\cite{Bavelas:HumanOrg1948,Sabidussi:psychometrika1966,Freeman:sn1979,BlochJackson:2007,CoKa:jcss07,Opsahl:2010,ECohenADS:TKDE2015,binaryinfluence:CIKM2014},
and influence
\cite{Gomez-RodriguezBS:ICML2011,ACKP:KDD2013,DSGZ:nips2013,timedinfluence:2014}
measures.  
 Often, however,  the relation of a node $j$ to $i$ is more 
   correctly modeled not by the magnitude of the distance $d_{ji}$
   from $j$ to $i$, but by $i$'s 
 position $\pi_{ji}$ in an ordering of nodes according to increasing 
    distance  from $j$ 
\cite{CoverHartkNN:InfoTheory1967,hastietibshirani:book2001,Wsabie:IJCAI2011,GaoJOW:KDD2015}. 
A classic use of rank as an indicator of relevance in 
metric spaces is the 
$k$ nearest neighbors (kNN) classifier, which classifies points based 
on the $k$ closest labeled examples 
\cite{CoverHartkNN:InfoTheory1967,hastietibshirani:book2001}. In terms 
of popularity, kNN outweighs 
the respective  distance-based classifiers, which instead use all examples within a 
certain distance.

More formally, we view a node $j$ as  {\em ranking} other nodes according to 
their distance order from  $j$.   The {\em rank} $\pi_{ji}$ is the position
of $i$ in increasing order from $j$.\footnote{$\pi_{ji}$  is also
  termed the {\em Dijkstra rank} of $i$, since Dijkstra's 
algorithm from source $j$ processes nodes in increasing distance.}
Accordingly, from the perspective of node $j$, we can refer to
$\pi_{ij}$ as a {\em reverse rank}.

An advantage of using rank is that it provides a different signal than distance by
``factoring out''  the effects of uneven density.
This is illustrated in the toy social network in 
Figure~\ref{network:fig}: 
We expect node $A$ to be more important to node $C$ than it is to node
$B$, even though, $A$ is closer to $B$ than to $C$ ($d_{CA} >
d_{BA}$).
This is because $B$ has a dense neighborhood of closer node than $A$,
but $C$ has only two nodes closer to it than $A$.
In terms of ranks, we have 
$\pi_{CA}=3$ and $\pi_{BA} = 6$ and thus $\pi_{CA} < \pi_{BA}$, which
reflects this intuition.

The rank relation is asymmetric:
In the example network in Figure 
\ref{network:fig}, we have 
$\pi_{BA} = 6$, since there are 5 nodes closer to  
$B$ than  
$A$,  and  $\pi_{AB}= 1$,  since $B$ is the closest node to $A$.
Therefore, 
$\pi_{AB} \not= \pi_{BA}$ even though the distance is
symmetric ($d_{AB}=d_{BA}$).  
The asymmetry $\pi_{BA} > \pi_{AB}$ reflects our intuition
that the higher degree node ($B$) has more influence on its neighbor
($A$) than the reverse.

In particular, with
tie breaking on distances, a node $v$ has exactly one nearest neighbor, but can have $0$ to many
{\em reverse nearest neighbors}, which 
are nodes $u$ which satisfy $\pi_{uv}=1$.  In our example, node $A$ has
no reverse nearest neighbors.  The number of reverse nearest neighbors of a point $v$
is a well studied notion of $v$'s influence, 
proposed by Korn and Muthukrishnan \cite{KornMuthu:sigmod2000}, and considered in 
metric spaces and in graphs \cite{YPMT:TKDE2006}.

\begin{figure}
\centering
\includegraphics[width=0.35\textwidth]{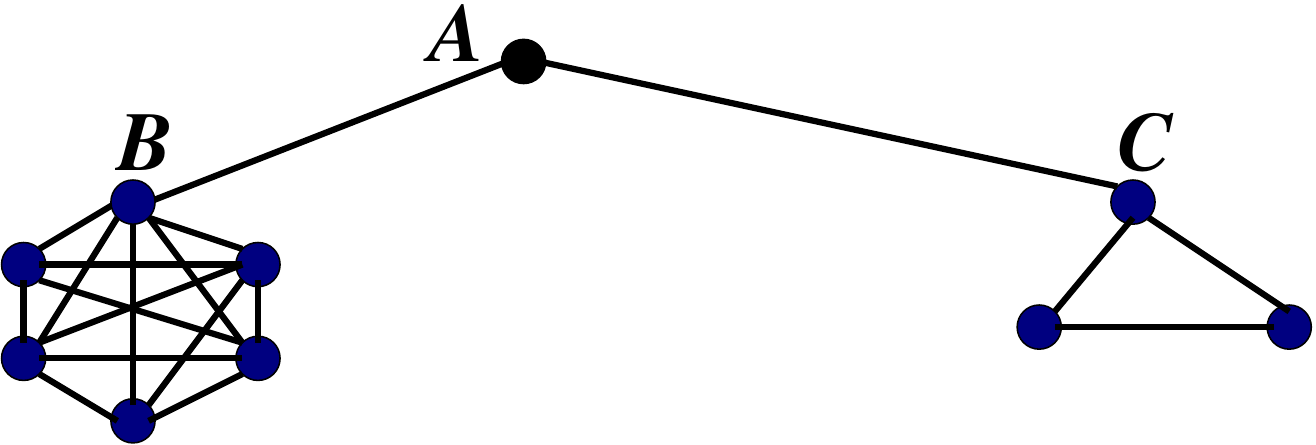}
\caption{\label{network:fig}{Example undirected social network (edge
    lengths are proportional to drawn lengths).}}
\end{figure}


In the basic model, which is sometimes called {\em monochromatic}  \cite{KornMuthu:sigmod2000},
all nodes both rank and get ranked.
 A natural extension ({\em bichromatic} model \cite{KornMuthu:sigmod2000}) allows only a subset of the 
 nodes to get ranked ({\em rankees}) and also permits  nodes that provide 
 ranks ({\em rankers}) to have different weights. In this model,
$\pi_{ij}$ relates 
a ranker $i$ to rankee $j$.
This distinction is useful when nodes have two or 
more types of entities, for example, users (rankers) and content (rankees).
 It also allows us to specify a special small
set of certified rankees such that we can characterize
properties of other nodes by this smaller set of ranks.
Importance weights $\beta(i) \geq 0$ assigned to rankers can
correspond to properties like
 purchase power or trust level. The ranks assigned by ranker $i$ are
 then weighted by $\beta(i)$.
This weighting is useful when we aggregate the scores of multiple
rankers to obtain centrality/influence scores of rankees.

\subsection{Contributions and Overview}
Rank-based measures provide a natural alternative to distance-based
ones, but algorithmically pose different challenges.
We identify and motivate fundamental challenges 
and present scalable algorithmic tools which
facilitating rank-based graph
mining.

\subsection*{Reverse-rank Single-source computation}
An important tool in working with distances is
an efficient single-source computation.
Dijkstra's algorithm from a source $s$ computes the 
  distances $d_{si}$ for all nodes $i$ in near-linear time.  A 
  powerful property of Dijkstra's algorithm is  {\em sorted access}:
Nodes are revealed in order of increasing distance from $s$.  Thus,  for 
  any $k$, the $k$ closest nodes to $s$ are computed while traversing 
  only edges adjacent to these $k$ nodes.   Therefore, if we are only interested
  in a prefix of the closest nodes, we can terminate the execution
  after they are computed, performing a fraction of the computation of
  a full execution.  When we work with ranks,  we
would be instead interested in a prefix of the highest ranks. Such
sorted access to rankings is also important for efficient
aggregating rankings \cite{FLN:jcss2003}.

 The  {\em reverse-rank single-source problem} is,
for a node $i\in U$,  to compute the reverse ranks $\pi_{ji}$ with 
respect to all nodes 
$j$.  Moreover, we aim for an efficient  algorithm that provides
sorted-access: Listing nodes in  
increasing $\pi_{ji}$ order with an algorithm that only traverses edges adjacent to 
listed nodes.

A naive solution for exact reverse-rank single source computation from $i$ is 
to run Dijkstra's algorithm from each node 
   $j$, until node $i$ is processed.  For the average node, this is 
   equivalent to performing $n$ runs of Dijkstra until revealing on 
   average $n/2$ nodes.  Note that even on sparse networks, this scales quadratically in the 
   number of nodes, which is 
   prohibitive even on mid-size networks.  This is in sharp 
   contrast to the shortest-path single-source computation which takes 
   (near) linear time. 
Previous work \cite{YPMT:TKDE2006} proposed ways to scalably identify 
the set of reverse nearest neighbors of nodes, but did not address
higher reverse ranks. We are able here  (Section 
\ref{hardness:sec}) to provide 
an explanation, establishing that the naive solution is in a sense the 
best we can do for the exact problem:
 We leverage the theory of subcubic equivalence 
\cite{WilliamsW:FOCS10} and construct a reduction from graph radius
computation to reverse-rank single source computation.   The former is
known to have subcubic equivalence to all-pairs shortest-paths 
computation (APSP) \cite{AGV:SODA2015}.  

 This hardness result, fortunately, applies to the {\em exact} problem. 
An important contribution we make here (see Section \ref{SSrDR:sec}) is devising a novel, scalable,
  Dijkstra-like (sorted access), 
  {\em  approximate}   reverse-rank single-source algorithm, which
  provides estimates $\hat{\pi}_{ij}$ with a small relative error.
Since ranks are intrinsically slightly noisy measures of the
actual relations,  estimates with a small relative errors are often as good
as the exact values. 

  An essential component of our design is a
preprocessing step where we compute All-Distances Sketches (ADS) for all 
nodes \cite{ECohen6f,bottomk07:ds,ECohenADS:TKDE2015}.  The sketches 
provide us with a fast oracle which estimates $\hat{\pi}_{ij}$ from 
the distance $d_{ji}$.
In Section \ref{ADS:sec} we review the sketches and estimators as
applied in our context.  We note that we can apply any ADS algorithm,
and existing designs are suitable for  
sequential, shared-memory, and  node-centric message-passing computations
  \cite{ECohen6f,PGF_ANF:KDD2002,bottomk07:ds,hyperANF:www2011,ECohenADS:TKDE2015}.
A stand-alone  contribution we make here is
engineering an ADS algorithm for multicore architectures which
provides provable tunable tradeoff between overhead and concurrency.   
Our algorithm can be used for many 
other applications of the sketches which include estimating distances, 
closeness similarity, the distance distribution, and timed-influence
\cite{ECohen6f,PGF_ANF:KDD2002,bottomk07:ds,hyperANF:www2011,
CDFGGW:COSN2013,DSGZ:nips2013,timedinfluence:2014,ECohenADS:TKDE2015}.

  The estimation of single-source  reverse ranks can be done by first running
  Dijkstra's algorithm  to compute the single-source distances, and
  then apply the oracle we obtained
  in the preprocessing step to the computed distances.    This method, however, will not provide us
sorted access.    Our sorted access algorithm, 
similarly to Dijkstra,  also traverses a shortest-path tree rooted at the source,
but critically, instead of doing so in distance order, which would
violate rank-based sorted access, does so in order of
increasing estimated reverse ranks.  The correctness of our design
relies on key  insights on properties of reverse ranks.

\subsection*{Reverse-rank influence}
Distance or reachability-based notions of centrality and influence of a set $S$ of seed nodes
are fundamental measures in network analysis.
In the general form \cite{timedinfluence:2014}, 
distance-based influence is defined with
respect to a
non-increasing decay function $\alpha(x) \geq 0$ (smoothing kernel) and node 
weights $\beta(i) \geq 0$.
The contribution of each node $j$ to the influence of $S$ is 
proportional to its weight $\beta(j)$ and decays with the distance of 
$j$ from $S$, $d_{Sj}=\min_{i\in S} d_{ij}$:
\begin{equation}\label{distinf:eq}
\textstyle{\Inf^{(d)}}(S) = \sum_j \beta(j) \alpha(d_{Sj})\ .
\end{equation}
Well studied special cases include
{\em Closeness centrality}
\cite{Bavelas:HumanOrg1948,Sabidussi:psychometrika1966,Freeman:sn1979,BlochJackson:2007,CoKa:jcss07,Opsahl:2010},
where $S$ contains a single node $i$, and
the celebrated {\em reachability-based} influence model of 
\cite{KKT:KDD2003}, obtained when 
$\alpha(x)=1$ for finite $x$ and $0$ otherwise.
Distance-based influence with threshold function $\alpha$
($\alpha(x)=1$ for $x\leq T$ and $\alpha(x)=0$ otherwise) was
studied in 
\cite{Gomez-RodriguezBS:ICML2011,ACKP:KDD2013,DSGZ:nips2013}
(With distance interpreted as elapsed time).\footnote{Reachability and 
  distance-based influence were also explored as the expectation
  when  edge lengths/presence are probabilistic.}

 Here we define  {\em reverse-rank influence} 
\begin{equation}\label{inf:eq}
\textstyle{\Inf^{({\pi^{-1}})}}(S) = \sum_j \beta(j) \alpha(\pi_{jS})\ ,
\end{equation}
where  $\pi_{jS}=\min_{i\in S} \pi_{ji}$. 
The special case of $\Inf^{({\pi^{-1}})}(v)$, the influence of a 
single node, with $\alpha(1)=1$ and $\alpha(x)=0$ otherwise 
is the number of reverse nearest neighbors of $v$, is the
influence measure proposed in \cite{KornMuthu:sigmod2000,YPMT:TKDE2006}.  
Our more flexible definition \eqref{inf:eq} is able to  account for the contribution of 
nodes with higher reverse-rank to the influence of our node.  For example,
by setting $\alpha(x)=1/x$ we achieve the effect that a reverse rank 
of $x$ contributes $1/x$ to the total influence of $v$; A node for 
which $v$ is the $5$th closest neighbor contributes to its influence 
$20\%$ of what it would have contributed as a reverse nearest
neighbor.  With $\alpha$ being a 
$T$-threshold function, rankers $u$ that rank $v$ in their top $T$
contribute $\beta(u)$ to $v$'s influence.  

\medskip
\noindent
{\bf Reverse-rank Influence Computation:}
  We show (Section  \ref{hardness:sec}) that the computation of exact reverse-rank influence, even
  for a single node, and even when $\alpha$ is a threshold function,
  has subcubic equivalence to APSP.  We therefore consider approximate
  influence $\widehat{\Inf}$ computed using approximate ranks
  $\hat{\pi}$.   Clearly $\widehat{\Inf}(S)$  can be computed using
$|S|$ single-source approximate reverse-rank computations (and using
$\hat{\pi}_{jS}=\min_{i\in S} \hat{\pi}_{ji}$). Surprisingly, however,
we show in Section \ref{infmax:sec}, that even with
large $|S|$, one single-source computation suffices.
\ignore{
For motivation, 
consider a network where
nodes correspond to user or movie entities.
The distance order of movies (rankees) from a user (ranker) 
can be interpreted as a ranking order of movies induced by the user.
The reverse-rank centrality of a movie $u$ is the number of users $v$ for
which the movie is in their top $T$.
The reverse-rank influence of a set of movies $S$ is then the number of
users which have at least one movie in $S$ in their top $T$ choices.
}

\medskip
\noindent
{\bf Reverse-rank Influence Maximization:}
An important coverage problem which is extensively explored
for reachability and distance-based influence, is {\em 
  influence maximization (IM)} \cite{KKT:KDD2003}:
For a given $s\geq 1$, identify a set of $s$ {\em seed}
nodes with maximum influence. 
Intuitively, such a set provides the best 
``coverage'' for its size with respect to the influence measure at
hand.  Here we consider IM with 
respect to our  reverse-rank influence function
$\Inf^{({\pi^{-1}})}$. 
The reverse-rank IM problem with $\alpha$ being a 
threshold function with parameter $T$ on our example user and movies 
data set is to find a set of $s$ movies  
which maximizes the number of users for which there is at 
least one movie from $S$  in their top $T$ choices. 

Similar to the distance-based influence function $\Inf^{(d)}$, $\Inf^{({\pi^{-1}})}$ is monotone and 
submodular, and even for simple threshold $\alpha$, when $s$
is a parameter,  the IM problem is NP hard. 
The most common and hugely successful algorithm for such coverage
problems is the greedy algorithm \cite{submodularGreedy:1978}, which iteratively
builds a seed set by selecting in each step a node with
maximum marginal 
contribution.  For submodular and monotone functions, greedy
 has the property that each prefix of the sequence of size $s$ has 
influence that is at least $1-(1-1/s)^s \geq 1-1/e$ of  the influence
of the optimal seed set of that 
size \cite{submodularGreedy:1978}.   Exact greedy, however, does not scale well for very large graphs.
For reverse-rank influence, an exact greedy sequence can be computed in cubic time in the number of 
nodes. 
When all nodes are both rankers and rankees, the graph is sparse, and 
we work with a threshold function $\alpha$ with parameter $T$, the 
computation reduces to $O(nT)$, by performing a single-source search 
from all nodes to find the $T$ nearest neighbors and computing a 
greedy cover.  But even this special case does not scale well for 
large graphs for larger values of $T$. 

 Approximate greedy and heuristics had been extensively studied for
reachability-based
\cite{Leskovec:KDD2007,CELFpp:WWW2011,TXS:sigmod2014,binaryinfluence:CIKM2014}
and distance-based \cite{DSGZ:nips2013,timedinfluence:2014} influence.
In particular, the \skim\ algorithm
\cite{binaryinfluence:CIKM2014,timedinfluence:2014}  computes in
near-linear time a full greedy permutation so that each prefix of size $s$
has approximation ratio of $1-(1-1/s)^s-\epsilon$.

  In Section \ref{infmax:sec}  we present a near-linear algorithm which computes,
an approximate  greedy sequence with
  respect to the {\em approximate} reverse-rank influence objective
  with threshold function:
$\widehat{\Inf}(S) = \{|z\in Z \mid \hat{\pi}_{zS}\leq T |\}$.    The algorithm we
  present builds on \skim, but incorporates critical adjustments  that utilize
the sorted access property of  our approximate reverse-rank single-source computations.


\subsection*{Experiments}

  Our experimental evaluation, detailed in Section
  \ref{experiments:sec}  was focused on scalability and solution
  quality, using publicly available anonymized social graph data
  sets.   Our ADS implementation runs on graphs with
  tens of millions of edges in tens of minutes on a single core,
  providing estimates with NRMSE (normalized mean square errors) of 
  6\%-13\%.   Our multithreaded design achieved speedup factors of
 3 to 4 on a machine with two CPUs and multiple cores.

  With the preprocessing in place,  our
approximate reverse-rank single-source computations have similar running
time to Dijkstra's algorithm (which computes single-source distances).
In particular, a reverse-rank single-source computation was performed in less
than 15 seconds on a single core on a graph with 4$\times10^6$ nodes
and 35$\times 10^6$ edges.   For comparison, this should be contrasted
with the running time of an exact reverse-rank
single source computation, which would have taken an estimated 6000 hours on
the same instance.

  Using our implementation, we are able to visualize the  reverse-rank distributions of some nodes in a
  large network, demonstrating how the distribution reveals
  information on the relative importance of a node in its locality.
 Prior to our work, it  was not possible to scalably compute these
 distributions on large graphs.

  Our approximate greedy IM implementation computes the full sequence
  on graphs with tens of millions of edges in minutes.  We also
  observe that for small graphs or small values of $T$, where we could
  compute an exact greedy sequence, the solution quality of our
  approximate sequence was very close to the exact one.

 \section{Preliminaries} \label{prelim:sec}
  We introduce some necessary notation. 
For a numeric function $r:X$ over a set $X$,  the function $\kth_r(X)$ returns the $\kth$ smallest value in the
range of $r$ on $X$.  If $|X|<k$,  we define $\kth_r(X)$ as the
supremum in the range of $r$.  If $r$ is not specified, we return the
$\kth$ smallest value in $X$.

We work with networks modeled as directed or undirected  graphs
$G=(V,E)$ with nodes $V= [n]=\{1,\ldots,n\}$ and edges $E$ with
lengths $w(e)>0$.   We use $m=|E|$ for the number of edges.
A subset or all nodes $U\subset V$ are specified as {\em rankee}
nodes.
We use the notation $G^T$ for
the {\em transpose graph}, which is the graph with edges reversed. 

For nodes $i,j$, let $d_{ij}$ be the shortest-paths distance from $i$ to $j$.
For $y\geq 0$, the {\em rankee $y$-neighborhood of $i$}
is the set of rankee nodes within distance $y$ from $i$.  We denote
the neighborhood by $$N_i(y)=\{j\in U \mid d_{ij} \leq y\}$$ and
its cardinality by $n_i(y)=|N_i(y)|$.
We use the notation $\underline{N}_i(y) = \{j\in U \mid d_{ij} < y\}$
for the respective strict neighborhood and $\underline{n}_i(y)$ for
its cardinality.
For $i\in V$ and $j\in U$,  $\pi_{ij}$ denoted the {\em rank} of $j$ with
 respect to $i$.  When distances are unique, we have $\pi_{ij} =
 n_i(d_{ij})$, that is, equal to the number of rankee nodes
that are at least as closer to $i$ as $j$.
 When distances are not unique, we consider the range
$(\underline{\pi}_{ij}, \overline{\pi}_{ij}]$, where
\begin{eqnarray}
\underline{\pi}_{ij} &=& \underline{n}_i(d_{ij})\ ,\label{ntopi}\\
\overline{\pi}_{ij} &=& n_i(d_{ij})\nonumber\ .
\end{eqnarray}
According to what we want to capture, we can define the rank 
$\pi_{ij}$ as  either,
$\overline{\pi}_{ij}$,  $\underline{\pi}_{ij}+1$, a uniform at random choice from
the range, or as the midpoint of this range: 
$\pi_{ij}\equiv \frac{\underline{\pi}_{ij}+1+\overline{\pi}_{ij}}{2}\
.$
Our algorithms and implementation can be adapted to support all these choices.


  An important ingredient of our design is the computation of a data
  structure which allows us to efficiently estimate the number of
  rankees in a neighborhood.
 That is, for a query specified by a node $i$ and
  $d\geq 0$, return $\hat{n}_i(d)$.
 The data structure can be viewed as a set of lists $L(i)$, one for each
 node $i\in V$.  Each list $L(i)$ consists of pairs $(d,y)$ where $d$
 is a distance value and $y=\hat{n}_i(d) > 0$ is an estimate on $n_i(d)$.
The lists are sorted and increasing in both $d$ and $y$.
To estimate for $n_i(x)$ and $\underline{n}_i(x)$, from the list $L(i)$, we use
\begin{eqnarray}
\hat{n}_i(d) &=& \text{$y$ such that  }\ (x,y) = \arg\max_{(x,y)\in L(i) \mid x\leq
  d} x\   \label{estfromlist:eq}\\
\hat{\underline{n}}_i(d) &=& \text{$y$ such that  }\ (x,y) =
                             \arg\max_{(x,y)\in L(i) \mid x <
  d} x\ . \label{lowerestfromlist:eq}
\end{eqnarray}
That is, we look at the pair $(x,y) \in L(i)$ such that $x \leq d$ (or $x<d$)
is maximum and return $y$.
From the relations \eqref{ntopi}, we can obtain estimates 
$\hat{\overline{\pi}}_{ij} = \hat{n}_i(d_{ij})$ 
and $\hat{\underline{\pi}}_{ij}= \hat{\underline{n}}_i(d_{ij})$ from $L(i)$ if we
know $d_{ij}$.

The lists $L(i)$ are computed from All-Distances Sketches $\ADS(i)$, which are
the subject of the next section.

\section{All-Distances Sketches} \label{ADS:sec}
 
 We preprocess the graph to 
compute a set of All-Distances Sketches (ADS)
  \cite{ECohen6f,ECohenADS:TKDE2015} for the nodes in the graph.  
The sketches are defined with respect to a 
  parameter $k$ and a random permutation of rankee nodes.  We
  find it convenient to work with $r(i)\in [0,1]$ which is the permutation
  position of $i$ divided by $|U|$.  Alternatively, it is  
sometimes convenient to work instead with random hash based
 $r(i) \sim U[0,1]$.
The sketch $\ADS(i)$ of a node 
 $i\in V$ consists of a set of entries of the form $(j,d_{ij})$, consisting
 of a node $j\in U$ and the distance $d_{ij}$.  We assume that 
 $r(j)$ is either included in the entry or can be easily retrieved
 from $j$.  The set of rankee nodes included in $\ADS(i)$ is a random
 variable which depends on the assignment $r$:
\begin{equation} \label{botkADS}
j \in \ADS(i) \iff  r(j) \leq \kth_r\{h\in U \mid d_{ih} \leq d_{ij}\}\ .
\end{equation}

This ADS definition \eqref{botkADS} applies with 
unique and non-unique distances. A technical point is that for estimation with non-unique distances, 
we also maintain with $\ADS(i)$, as auxiliary, entries 
$(j,d_{ij})$ that satisfy for some $z\in\ADS(i)$
$r(j) =\kth_r\{h\in U\setminus \{z\} \mid d_{ih} \leq d_{iz}\}$ when
these entries are not already included in $\ADS(i)$
\cite{multiobjective:2015}
(When distances are unique, all these entries are already in $\ADS(i)$). 
With unique distances, the expected size of the sketches is exactly
$\sum_{i=1}^{|U|} \min\{1,k/i\} \leq k\ln 
 |U|$  with good concentration, but the sketch can be
much smaller when distances are not unique, while providing the same
statistical guarantees on estimate quality, which is why we separately treat
non-unique distances rather than tie break.


Our implementation of ADS computation is based on
  {\sc Pruned Dijkstra's}
  \cite{ECohen6f,bottomk07:ds,ECohenADS:TKDE2015}.  The pseudocode for
  the basic sequential version is provided as
  Algorithm~\ref{pDijkstra:alg} and uses  $O(km\ln n)$ edge
  traversals.  When applied with non-unique distances, the algorithm
  also includes the auxiliary entries.

\begin{algorithm}[h]
\caption{$\ADS$ set for $G$ via {\sc Pruned
    Dijkstra's}\label{pDijkstra:alg}}
{\small
\For {rankee node $u\in U$ by increasing $r(u)$}
{
Run Dijkstra's algorithm from $u$ on $G^T$\\
\ForEach{scanned node $v$}
{
\eIf 
{$d_{vu}> \kth\{y \mid (x,y)\in \ADS(v)\}$ \Or  \\ $|\{x\in\ADS(v) \mid
  d_{vx} \leq d_{vu}\}| > k$}
{prune Dijkstra at
  $v$}{$\ADS(v)\gets \ADS(v) \cup \{(r(u),d_{vu})\}$ }
}
}
}
\end{algorithm}
The term {\em scanned node} in the pseudocode refers to the event
where the node $v\in V$
 is popped from the Dijkstra priority queue.  Each node can be scanned
 at most once in each (pruned) Dijkstra search.  The scanned nodes are
 always a prefix of the nodes when sorted by increasing distance from
 $u$ in $G^T$. When
 a node $v$ is scanned, either $u$ is inserted to $\ADS(v)$ or the search
 is pruned at $v$.  Therefore, the number of node scans is equal to
 the ADS size.

  The algorithm builds the $\ADS$ of all nodes 
by considering one node $u\in U$ at a time and adding it as an entry in
  $\ADS(v)$ for all relevant $v$.
  To do so efficiently, we maintain the entries in $\ADS(v)$ as an array sorted by
  decreasing distances.   The
insertion condition then amounts to testing
if  $|\ADS(v)|<k$ or if
 the entry $(x,y)$ in the $|\ADS(v)|-k$ position in the array (the
$k$th smallest distance) has $y> d_{vu}$, or if
it has $y=d_{vu}$ but either $|\ADS(v)|=k$ or 
the entry $(x,z)$ in the $|\ADS(v)|-k-1$ position has $z>d_{vu}$.

  We refer to the
$k$th smallest distance in $\ADS(v)$  as the {\em threshold distance}
and denote it by $\Delta(v)$.  We also use the notation $*(v)$ for the
bit indicating if the $k+1$th smallest is equal to the $k$th smallest distance.
The prune condition can then be written as
\begin{equation} \label{prune:cond}
d_{vu} > \Delta(v) \text{ \Or\ } d_{vu}=\Delta(v) \text{ \And\ } *(v)\ .
\end{equation}
Observe that insertions can only  affect the $k$ last entries in the
current $\ADS$.
 Therefore, it suffices to keep only that ``tail'' part in active
memory.  When $k$ is small we can keep it as an array and implement
insertions by shifting.  When $k$ is
larger we can use a data structure that supports efficient insertions.


\subsection{Multithreading} \label{multithread:sec}
  {\sc Pruned Dijkstra's}, as stated, sequentially performs possibly dependent searches
  from all rankee nodes.  We propose here a design  which allows us to 
control in a principled way 
 the tradeoff between overhead and concurrency. 
We partition the $|U|$ pruned Dijkstra searches to batches, where
each batch is a consecutive set of nodes when ordered by increasing $r$.
All the searches in the same batch are made independent so that they
can be executed concurrently.  Each search
computes a set of {\em proposed entries} to sketches $\ADS(v)$, as 
contributions to a set $PE(v)$.  A proposed entry is created when a
node $v$ is visited and the pruning condition \eqref{prune:cond} is
not satisfied.  The pruning, however, and
hence the proposed entries, are computed  with respect to the
set of threshold distances and bits $(\Delta(v),*(v))$ for $v\in V$, as it was at the beginning of the batch.  Pseudocode for
an independent search thread is provided as Algorithm \ref{search:alg}.
 Each such Dijkstra search 
may generate a proposed ADS entry for multiple nodes.

At the end of a batch, for each node $v$, the proposed entries
$PE(v)$ from all the searches in the batch are merged
with (the $k$-tail of) $\ADS(v)$ (as it was in the beginning of the
batch) to compute an updated $\ADS(v)$ with respect to the end of the
batch.  
The merge is performed by scanning the entries $(u,d_{vu})$ in 
$PE(v)$ in order of increasing $r$ and applying the insertion 
procedure to $\ADS(v)$ as used in Algorithm \ref{pDijkstra:alg}: If 
the pruning condition \ref{prune:cond} does not hold, we insert $u$
and update $\ADS(v)$ (note that this updates $\Delta(v)$ and $*(v)$).
Note that not all proposed entries are incorporated, since 
the insertion rule is not satisfied with respect to the updated
$(\Delta(v),*(v))$ after processing previous $PE(v)$ entries.

\ignore{
The merge is performed by sorting $PE(v)$ by increasing distances 
(breaking ties by smaller rank).
We then scan the sorted PE list, maintaining scanned entries sorted by rank.
The current entry will be included if the number of $\ADS(v)$ entries
of no-larger distance and the number of processed $PE(v)$ entries of
smaller rank are together less than $k$.
This merge is performed in $O(\min\{k, |PE(v)|\log k\}+ |PE(v)|\log|PE(v)|)$ time.
\notinproc{Pseudocode for this merge is provided as Algorithm \ref{merge:alg}.}
}

\begin{algorithm}[h]
\caption{A  {\sc Pruned Dijkstra} thread (search from
  $u$) \label{search:alg}}
{\small
Run Dijkstra's algorithm from $u$ on $G^T$ \\
\ForEach{scanned node $v$}
{\eIf {$d_{vu} > \Delta(v)$ \Or $d_{vu}=\Delta(v)$ \And $*(v)$}{prune
    Dijkstra at $v$}{$\, PE(v)\gets PE(v) \cup \{(u,d_{vu})\}$}}
}
\end{algorithm}

\subsection{Concurrency/Overhead tradeoff analysis}
The sequential algorithm has the property that all generated entries
constitute final ADS entries.  The
multithreading algorithm computes proposed entries that may be eventually discarded.  
These discarded entries are the overhead of 
the multithreading algorithm.

More precisely, we define the overhead as the 
ratio of the expected number of 
discarded entries to the expected number of ADS entries per node.  The overhead depends on
how we partition the searches to batches.  Placing each search in
a separate batch would result in no overhead, but also no
concurrency.  Putting all searches in the same batch would have a very
large overhead, as none of the searches would be pruned.

 Note that the overhead of discarded entries corresponds to an
 overhead on edge traversals, which are the main cost of the algorithm.
In particular, we can bound the total work performed by the
multithreaded algorithm by multiplying 
 the sequential bound of $km\ln |U|$ by
$(1+h)$, where $h$ is a bound on the overhead.

We next propose batch partitions which allow us to bound the overhead.
We first observe that the search is never pruned for the $k$
nodes with lowest $r$ values.  Our first batch would contain these
nodes, and we can perform those searches
independently without overhead.  At the end of this first batch, all generated proposed entries $PE(i)$ would be sorted by distance 
to form $\ADS(i)$ with respect to those $k$ nodes.  As for subsequent
batches, we propose exponentially increasing batch sizes 
and show the following:
\begin{lemma}
For a parameter $\mu>0$, consider a 
partition to batches so that the $j$th batch ends at node in
position
$\lceil (1+\mu)^{j-1} k \rceil$ in the sorted order by increasing $r(v)$.
Then the expected overhead  is at most 
$h \leq \mu/\ln(1+\mu)-1$. 
\end{lemma}
\begin{proof}
Consider processing  a batch that starts at position $b_0+1$.
 The probability
  of a node in the batch to enter  $PE(i)$ is $\min\{1,
  \frac{k}{b_0+1}\}$.  Note that to generate a proposed entry, the node needs to be with
  distance smaller than $\Delta(i)$, that is be among the $k$ smallest
  distances among all the nodes processed up to the previous batch and
  itself.
This probability is exactly that of being in one of the first $k$
positions in a random permutation of $b_0+1$ nodes, which is $\min\{1,
  \frac{k}{b_0+1}\}$.
Now we can consider the probability that a node in the batch is a
final member of $\ADS(i)$. 
 If the node is in position $b_0+j$, the
probability is 
 $\min\{1,
  \frac{k}{b_0+j}\}$. 

We now consider a batch that has nodes in permutation positions
$b_0+1$ to $b_t$,  such that  $b_0>k$.
The ratio of good work to total work is
\begin{eqnarray*}
\frac{\sum_{j=1}^{b_t-b_0} 1/(b_0+j)}{(b_t-b_0)
  /(b_0+1)} & = & \frac{b_0+1}{b_t-b_0}\sum_{j=1}^{b_t-b_0}
\frac{1}{b_0+j} \\ &\approx& \frac{b_0+1}{b_t-b_0}\ln(\frac{b_t}{b_0+1}) 
\end{eqnarray*}

Thus, if we choose $b_t= (1+\mu) b_0$ for some $\mu>0$, we obtain
the ratio $\ln(1+\mu)/\mu$.   The overhead, by definition,  is the inverse of this
ratio minus $1$.
\end{proof}
In particular, we can see that $\mu=0.5$ results in overhead of 
20\% more edge traversals than the sequential algorithm.  Using 
$\mu=0.1$, has overhead of about 5\%.
The total number of batches is
$O(\log_{1+\mu} (|U|/k)) \approx O(\frac{\log(|U|/k)}{\mu})$, which
is logarithmic in the number of rankee nodes.

\subsection{Cardinality estimation} \label{NCest:sec}
We now discuss the computation of a list $L(i)$ 
from $\ADS(i)$. 
Recall that  $L(i)$ is a list of pairs the form 
$(d, \hat{n}_i(d))$.  There is one  pair for each unique distances $d$ in
$\ADS(i)$, and we assume $L(i)$ is sorted by increasing $d$.

  We review two estimators $\hat{n}_i(d)$: The bottom-$k$ estimator
  and the HIP estimator.  Both are unbiased, nonnegative, and 
have a small relative error,  with good 
concentration which depend on the ADS parameter $k$.   
The HIP estimate is tighter: Estimates are at least as good as
bottom-$k$, and with unique distances, has half the
variance of  the bottom-$k$ estimator.  The bottom-$k$ estimator, however, is
useful to us because it has some monotonicity property.
We can compute the lists
$L(i)$ using  either estimator or both.

\smallskip
\noindent
{\bf The bottom-$k$ estimator:}
This inverse probability estimator \cite{HT52} has
coefficient of variation (CV) at most $1/\sqrt{k-2}$
\cite{ECohen6f,ECohenADS:TKDE2015}. 
  To estimate $n_i(x)$, we take the $k$th
  smallest $r$ value among nodes in $N_i(x)$, which we denote by $\tau$. 
 If there are fewer than $k$ nodes in $N_i(x)$, we return the number of
 entries as our estimate. Otherwise, we  compute
  the probability $p$ that an $r$-value is below $\tau$.  When $r(v) \sim
  U[0,1]$,   $p=\tau$.
We then use the estimate $\hat{n}_i(x)=(k-1)/p$.

\smallskip 
\noindent 
{\bf The HIP estimator:} This estimator has CV at most $1/\sqrt{k-2}$
and with unique distances is most $1/\sqrt{2k-2}$
\cite{ECohenADS:TKDE2015} (see \cite{multiobjective:2015} for 
extension to non-unique distances).  The estimates are obtained as follows:
For each (non auxiliary) entry $j$ in $\ADS(i)$, we compute the threshold value
\begin{equation} \label{botkRCthreshold}
\tau_{ij}=
\text{k}^{\text{th}}_r\{h\in \ADS(i) \mid d_{ih}<d_{ij}\}\ .
\end{equation}

We then compute $p_{ij}$ as the probability of $r(j)< \tau_{ij}$.  If
there are fewer than $k$ entries lower than $d_{ij}$ then $p_{ij}=1$.
Otherwise, when
  $r(j) \sim U[0,1]$, we have $p_{ij}=\tau_{ij}$.  
We then take $a_{ij} = 1/p_{ij}$.  Finally, the HIP estimate (summed
over non-auxiliary entries)  is
$$ \hat{n}_i(x) = \sum_{j\in \ADS(i) \mid d_{ij}\leq x} a_{ij} \ .$$

\ignore{
\smallskip 
\noindent
\begin{quote} 
{\bf To think about:  A batched HIP estimator:}  Can we brings some of the
HIP advantages to non-strict ADS ?
We compute the estimate by processing each tier (entries with same
distance) together.  We  process distance tiers by increasing $d$.
We then compute
the estimate by adding up contributions for each tier.    

Idea:  Suppose that for each tier we also have information on the smallest
  $r$ value that did not make it into the sample (we only need to include
  it if it is smaller than the $k$th largest in previous tiers.)

  Then can we can apply a conditioning estimator for each member of
  the tier?
members not sampled get estimate of $0$.  For sampled members we
compute the inclusion probability conditioned on tanks of all other
nodes being the same.

\end{quote}
}

\smallskip
\noindent
{\bf Computing the estimates}
The estimation list $L(i)$ for both the bottom-$k$ and the HIP estimators
  can be computed by processing the 
  entries of $\ADS(i)$ in increasing distance order, maintaining the 
  $k$th smallest values in the prefix processed so far and accordingly
  the $k$th smallest value $\tau$, and computing
  the estimates $\hat{n}_i(d)$ when entries of distance $d$ are
  processed.

 An easy to verify property that is useful to us is that
the neighborhood size estimates for each node are non-decreasing with
distance from the node and therefore can only increase when the distance does:
\begin{lemma} \label{monest:lemma}
When $L(i)$ is computed using either bottom-$k$ and HIP estimates,
the estimates \eqref{estfromlist:eq} and \eqref{lowerestfromlist:eq} satisfy
\begin{eqnarray*}
 d_1 \leq d_2 &\implies& \hat{n}_i(d_1) \leq \hat{n}_i(d_2) \\
 d_1 \leq d_2 &\implies& \hat{\underline{n}}_i(d_1) \leq \hat{\underline{n}}_i(d_2)
\ .
\end{eqnarray*}
\end{lemma}

\section{Reverse-Rank Single-Source} \label{SSrDR:sec}

As noted in the introduction, if we are  interested in computing reverse-ranks from a source $i$ to all nodes,
we can compute the distances $d_{ji}$ by applying Dijkstra's
algorithm from $i$ on $G^T$, and return  estimated reverse-ranks
from the distances using \eqref{estfromlist:eq} and
\eqref{lowerestfromlist:eq}.
The nodes, however,  are processed in order of increasing distance, which does not
necessarily corresponds to the order by increasing reverse ranks (recall the
example in the introduction).  
Therefore, if we are only
 interested in correctly identifying nodes with highest reverse ranks and we
 apply this algorithm, we can not prune the computation and we  will 
scan a much larger portion of the graph than needed.

In this section we present an approximate reverse-rank single-source 
algorithm that provides sorted-access:  Computing 
nodes in order of 
 increasing (approximate) reverse-rank.    
 We start by establishing a basic
 monotonicity property of
 reverse-ranks that is essential for the correctness of our design.

\subsection{SP monotonicity of reverse-ranks} \label{SPmonotonicity:sec}

We noted in the introduction that reverse-rank order does not
necessarily correspond to distance order.   For nodes on a shortest
path, however, we can show that the reverse-ranks, and the respective
bottom-$k$ estimates,  are monotone:
 \begin{lemma} \label{monotone:lemma}
Consider a shortest path $i_t,\ldots,i_0$ in $G$.
Then
$\overline{\pi}_{i_j   i_0}$, $\underline{\pi}_{i_j   i_0}$, and the
bottom-$k$ estimates
$\hat{\underline{\pi}}_{i_j i_0}$ and $\hat{\overline{\pi}}_{i_j
  i_0}$, are all
non-decreasing with $j$.
\ignore{
\begin{itemize}
\item 
 $\pi_{i_j 
  i_0}$  and the bottom-$k$ estimates $\hat{\pi}_{i_j i_0}$  (when 
distances are unique) 
\item 
The interval bounds 
$\underline{\pi}_{i_j i_0}$ and $\overline{\pi}_{i_j i_0}$ and the respective 
bottom-$k$ estimates 
 $\hat{\underline{\pi}}_{i_j i_0}$, $\hat{\overline{\pi}}_{i_j i_0}$ (when distances are 
 not unique) 
\end{itemize}
are  non-decreasing with $j$. 
}
\end{lemma}
\begin{proof}
Consider $j<h$, then $d_{i_j i_0} < d_{i_h i_0}$.   The neighborhoods
relations is $N_{i_j}(d_{i_j i_0}) \subseteq N_{i_h}(d_{i_h i_0})$.  Therefore,
$\overline{\pi}_{i_j i_0} \leq  \overline{\pi}_{i_h i_0}$. 
Similarly,
$\underline{N}_{i_j}(d_{i_j i_0}) \subseteq \underline{N}_{i_h}(d_{i_h
  i_0})$ and thus the cardinalities satisfy
$\underline{\pi}_{i_j i_0} \leq  \underline{\pi}_{i_h i_0}$. 

In the case of bottom-$k$ estimates, the claim follows again from
containment of neighborhoods.  
Let $\tau_1$ be the $k$th smallest $r$ value in the contained
 set and let $\tau_2$ be the $k$th smallest $r$ value in the
 containing set.  Then clearly, $\tau_1 \geq \tau_2$.
Recall that the bottom-$k$ cardinality estimate is
 $(k-1)/\tau$.  We have  $(k-1)/\tau_1 \leq (k-1)/\tau_2$.
\end{proof}

\notinproc{\footnote{
  Note that the path property may not hold for  HIP 
  estimates.  This is because the HIP estimates are sensitive to the distance order in which nodes
are scanned.  So even estimates obtained for the same neighborhood,
when scanned from two different nodes, will be different.
Thus, the approximate reverse-rank single source algorithm we
are about to present uses  bottom-$k$
estimates.
Only after that we can apply HIP to obtain tighter estimates.
}}

\subsection{Algorithm and analysis}
The pseudocode for our reverse-rank single source algorithms is
presented as   Algorithm~\ref{ssrdr:alg}.  The algorithm has the same
structure as  Dijkstra's algorithm from source $s$  in the transposed
graph $G^T$.
 The algorithm maintains each unprocessed node $j$ that is adjacent to an already processed node 
 in a min priority queue.  The entry contains $(j,\overline{d}_{js})$,
 where $\overline{d}_{js}$ is the upper bound on the distance from $j$
 to $s$.  This upper bound 
 serves as the priority in  Dijkstra's algorithm and 
is the minimum over processed nodes $h$ of $d_{hs}+w_{jh}$.   
Our reverse-rank single source algorithm 
uses instead a priority as follows.  We first look at 
$\hat{n}_{\overline{d}_{js}}(j)$, which is 
an upper bound on the estimated  reverse-rank, computed 
 according to the best current upper bound $\overline{d}_{js}$ on the
distance.  From Lemma~\ref{monest:lemma},
$\hat{n}_{\overline{d}_{js}}(j) \geq  \hat{n}_{d_{js}}(j)$. Therefore,
when the upper bounds on the distances are tightened, the priority
can only decrease.
Now, two nodes in the priority queue can have the same estimate $\hat{n}_{\overline{d}_{js}}(j)$.
 In this case, we break ties according to the distance
upper bounds $\overline{d}_{js}$,  always preferring the node with
lower $\overline{d}_{js}$.  If both $\hat{n}_{\overline{d}_{js}}(j)$
and $\overline{d}_{js}$ are the same, the tie can be broken arbitrarily.

  The next node $h$ that is selected from the queue is the one with minimum 
  priority according to lexicographic order on $(\hat{n}_{\overline{d}_{js}}(j),\overline{d}_{js})$.
 For this node $h$ we set $d_{hs} \gets \overline{d}_{hs}$ (a
 correctness proof that indeed $\overline{d}_{hs}$ is the distance is
 provided below).
  We then scan all in-coming edges $(j,h)$.  If $j$ is not already in
  the priority queue, we insert it with
  $\overline{d}_{ji}=d_{hs}+w(j,h)$ and the respective priority.  If
  $j$ is already in the queue we compare $x\gets d_{hs}+w(j,h)$ to
  the current $\overline{d}_{js}$.  If $x< \overline{d}_{js}$, we update
$\overline{d}_{js} \gets x$ and update the priority to
$\hat{n}_{\overline{d}_{js}}(j)$.

  Note that the algorithm applies for both directed and undirected
  graphs.  When applied to directed graphs, the algorithm returns
  reverse ranks only for
 nodes that can reach $s$.  For completeness, we 
explain how to extends this, if needed, also for nodes $v$
  that can not reach $s$, that is, $d_{vs}=\infty$.    We first need
  to precisely define the rank $\pi_{vs}$ in this case.  All rankee nodes that can
  not be reached from $s$ can be viewed as having rank range 
$(|R_v|, |U|]$, where $R_v$ is the set of rankee nodes reachable 
from $v$.  Now note that we can estimate $|R_v|$ by the 
cardinality estimate associated with the maximum-distance  entry in
$L(v)$.

\begin{algorithm2e}[h]
\caption{Approximate Reverse-rank Single-Source \label{ssrdr:alg}}
{\small
\DontPrintSemicolon
%
\SetKwData{E}{E}
\SetKwData{Q}{Q}
\SetKwFunction{output}{Output}
\SetKwData{s}{s}
\SetKwData{v}{v}
\SetKwData{u}{u}
\KwIn{Source node $\s$}
\KwOut{A sequence $(\v,d_{\s \v},\hat{\pi}_{\v \s})$ in increasing $\hat{\pi}_{\v \s}$}
\SetKwArray{dist}{dist}
\SetKwArray{rDra}{rDr}
\SetKwFunction{addQ}{add}
\SetKwFunction{Init}{Init}
\SetKwFunction{extractmin}{extract\_min}
\SetKwFunction{decreasekey}{decrease\_key}
\tcp{Node object $\v$: $\v.\dist$,$\v.\rDra$ are upper bounds on
  $d_{\s \v}$ and $\hat{\pi}_{\v \s}$. Initialization $\v.\Init(d,r)$: $\v.\dist \gets d$; $\v.\rDra \gets r$}
$\Q \gets \emptyset$ \tcp*{Initialize an empty min priority queue $\Q$ of node objects
prioritized by lex order of $(\v.\rDra,\v.\dist)$}
$\s.\Init(1,0)$ \tcp*{Initialize source node object}
$\Q.\addQ{\s}$\tcp*{put source node in queue}

\While{$\Q$ is not empty}
{
$\v \gets \Q.\extractmin{}$ \;
$\output{\v}$ \tcp*{output and scan $\v$}
\ForEach{$\u \mid (\u,\v)\in \E$ and $\u$ not scanned}
{
  $d\gets \v.\dist+w_{vu}$\;

  \eIf{$\u \not\in \Q$}
  {
    $\u.\Init(\hat{n}_{\u}(d),d)$\; 
    $\Q.\addQ(\u)$
  }
  {
    \If{$d < \u.\dist$}
    {
$\Q.\decreasekey(u, (\hat{n}_{\u}(d),d))$
      \tcp*{Update priority of $\u$
        in $\Q$: $\u.\dist \gets d$; 
      $\u.\rDra \gets \hat{n}_{\u}(d)$ }
    }
  }
 }
}
}
\end{algorithm2e}


For correctness, we need to show that when a node $v$ is popped out
from the priority queue, we have the correct distance $d_{sv}$ and
thus can obtain the bottom-$k$ estimate on $\pi_{vs}$.  This holds
if all nodes that are on the shortest path from $v$ to $s$ were
scanned before $v$. 
\begin{theorem}
When Algorithm \ref{ssrdr:alg}  is applied with
exact cardinalities $n_i(d)$ or with
bottom-$k$ estimates, it
traverses a shortest-paths tree from $s$.
\end{theorem}
\begin{proof}
 Consider a source $s=y_0$ and let $y_0,y_1,\ldots$ nodes sorted by increasing
$\hat{\pi}_{y_i s}$ with ties broken according to $d_{y_i s}$.
We show that a node can be scanned 
only after all the nodes on its shortest path to $s$ are also scanned.  
This means that when scanned, its 
 current priority is computed according to its true distance to $s$,
 and therefore, uses the bottom-$k$  reverse-rank estimate.

We show correctness by induction on $t$.  Assume that the scanned
nodes are $y_0,\ldots,y_t$ and that for these nodes we have 
exact SP distances and thus the estimates
$\hat{\pi}_{y_i s}$.
 Consider now $y_{t+1}$. Consider the shortest path $P$ from $y_{t+1}$
 to $s$. It follows from Lemma \ref{monotone:lemma} that the reverse-rank
 estimates are
 monotone non-decreasing along the path. Also note that distances to
 $s$ are strictly smaller.   Therefore, all the nodes of path
$P$, except $y_{t+1}$,  are in $\{y_0,\ldots,y_t\}$,  and therefore, by induction
 already scanned.  
\end{proof}

\section{Reverse-rank Influence} \label{infmax:sec}

\SetKw{Skip}{skip}
\SetKw{Continue}{Continue}
\SetKw{Prune}{prune}
\SetKw{Return}{return}
\SetKw{Or}{or}
\SetKw{And}{and}
\SetKwFunction{Append}{append}
\SetKwArray{Wasdelta}{best\_seed}
\SetKwArray{Index}{inverted\_sample}
\SetKwArray{Size}{sample\_size}
\SetKwArray{SeedList}{seedlist}
\SetKwData{Rank}{rank}
\SetKwFunction{Sort}{sort}
\SetKw{Skip}{skip}
\SetKw{Return}{return}
\SetKwFunction{Insert}{insert}
\SetKwFunction{Merge}{merge}
\SetKwData{Ads}{ADS}

\SetKwArray{Coverers}{coverers}
\SetKwArray{Coverage}{coverage}

 In this section, we consider the computation and maximization of
reverse-rank influence.  Consider a graph with a set of rankee nodes
$U\subset V$ and ranks $\pi_{ji}$ defined for rankees $i\in U$ and
$j\in V$.   Let $\beta(j)\geq 0$ be the ranker weights of $j\in V$.
For a set $S \subset U$ of seed nodes, the
reverse-rank influence is
$\Inf(S)=\sum_{j\in Z} \beta(j)\alpha(\pi_{jS})$, where
$Z\subset V$ is the set of ranker nodes (those with $\beta(z)>0$).
From Corollary \ref{hardnessinfluence:coro}, the exact computation of
$\Inf(S)$
has subcubic equivalence to APSP, even 
when restricted to threshold functions $\alpha$ and a single seed.

We therefore focus on scalably computing the approximate influence
$$\widehat{\Inf}(S)=\sum_{j\in Z} \beta(j)\alpha(\hat{\pi}_{jS})\ ,$$
where $\hat{\pi}_{jS} = \min_{i\in S} \hat{\pi}_{ji}$.

Note that to compute $\widehat{\Inf}(S)$ it suffices to compute
$\hat{\pi}_{jS}$ for all ranker nodes $j\in Z$.
Moreover, when $\alpha$ is a threshold function for some $T\ll n$, or
more generally,  any  function with 
$\alpha(x)=0$ for all $x>T$, it suffices to
compute $\hat{\pi}_{jS}$ only for nodes with $\hat{\pi}_{jS} \leq T$.
A naive way to compute these values is to perform, from each seed
$i\in S$,  a single-source
reverse ranks search from $i$, using Algorithm~\ref{ssrdr:alg},
and terminate the search when we scan a node with $\hat{\pi}_{ji} > T$.
We can then combine the results of the different searches, computing
the minimum $\hat{\pi}_{ji}$ of each node $j$ that is scanned in at
least one of the searches, to obtain the values $\hat{\pi}_{jS}$.
This naive computation requires $|S| |E|\log n$ operations (assuming the
lists $L(j)$ are provided) when $T$ is large.  But even with smaller
$T$, a node $j$ can be scanned multiple times, once for each seed  $i\in S$ with
$\hat{\pi}_{ji} \leq T$.
We now show how to remove the dependence on the number of seeds $|S|$.
\begin{theorem}
For a set of seeds $S$, we can compute the values $\hat{\pi}_{jS}$ for
all $j\in V$ using 
$O(|S|+|E|\log n)$ operations.  When $\alpha(x)=0$ for $x\geq T$,
$|E|$ is replaced by the number of incoming edges to nodes $j$ that
satisfy $\hat{\pi}_{jS} \leq T$.
These bounds assume that the lists $L(j)$ are provided for $j\in V$.
\end{theorem}
\begin{proof}
 We slightly modify Algorithm
 \ref{ssrdr:alg} by initializing the priority queue  with entries with priorities
 $(i.\dist,i.\rDra)= (0,1)$ for 
each node $i\in S$.  The algorithm execution then  proceeds as with a
single source node.
For correctness, we can show that nodes are scanned (popped from the queue) in increasing lex
order of $(\hat{n}_j(d_{jS}),d_{jS})$, and at the point they are scanned we
have
$j.\dist = d_{jS}$ and thus $j.\rDra = \hat{n}_j(d_{jS})$.

To see that, first note that $d_{jS}$ suffices to obtain
$\hat{\pi}_{jS}$.  This is because, using Lemma~\ref{monest:lemma},
$\hat{n}_j(d_{jS}) \equiv \min_{i\in S} \hat{n}_j(d_{ji}) =
\hat{n}_j(\min_{i\in S} d_{ji})$.  

 For correctness, we need to show that the monotonicity property (Lemma
\ref{monotone:lemma}) to sets $S$:
Consider a shortest path $i_t,\ldots,i_0$ from $i_t$ to $S$ (that is,  $i_0$ is the
closest $S$ node to $i_t$).  Note that this implies that $i_0$ is the
closest $S$ node to all $i_j$.

It follows from \ref{monotone:lemma} 
that $\overline{\pi}_{i_j   i_0}$, $\underline{\pi}_{i_j   i_0}$, and the
bottom-$k$ estimates
$\hat{\underline{\pi}}_{i_j i_0}$ and $\hat{\overline{\pi}}_{i_j
  i_0}$ are non-decreasing with $j$.  
We  observe that $\overline{\pi}_{i_j   i_0} = \overline{\pi}_{i_j 
  S}$ and similarly $\underline{\pi}_{i_j   i_0} =
\underline{\pi}_{i_j   S}$ and this holds for the bottom-$k$ estimates 
obtained using $d_{i_j S}$.  Therefore, monotonicity holds also when
$i_0$ is substituted with $S$.
\end{proof}

\subsection{Influence Maximization}

 Here we consider uniform ranker weights $\beta(z)=1$ and 
 $\alpha$ that is a threshold function for some $T$.   
The influence of a set $S\subset U$ of rankee nodes is then 
the number of rankers that have at least one node from $S$ among their 
top $T$ rankees:
\begin{equation}  \label{rrInfT}
\Inf(S) = |\{  z\in Z \mid \pi_{zS} \leq T\}|\ .
\end{equation}

The goal of the IM problem is to find a set $S$ of rankee nodes
of a certain size which maximizes $\Inf(S)$.
A common approach to such coverage problems is the greedy algorithm.
Greedy repeatedly selects a rankee node which has maximum marginal influence.
For each $s \geq 1$, the set $S$ of the first $s$ selected seeds is guaranteed
to have influence that is at least $1-(1-1/s)^s$ of the maximum
possible by $s$ seeds.
Algorithm \ref{exactgreedyim:alg} is an
exact greedy algorithm for our reverse-rank IM problem with influence
function \eqref{rrInfT}.
 The computation of the
algorithm is dominated by Dijkstra computations from each ranker node that
are stopped once $T$ rankees are popped from the priority queue.  Recall
that when distance are not unique, we can work with multiple
definition of the rank (see Section \ref{prelim:sec}), but with all of them, we
can determine the ranks once at most $T+1$ rankees are popped.  
example, when we use $\overline{\pi}$, then
if the
$T+1$ rankee has the same distance $d$ as the $T$ rankee, then all
rankees of distance $d$ are excluded (they have rank larger than $T$).
 Even when all nodes are rankees, and thus at most $T$ nodes are
 popped in total in each Dijkstra run,  the required
computation is $\Omega(T |E| \log |V|)$, which does not scale well for
large values of the threshold $T$.

\begin{algorithm2e}[!h]
\caption{Exact greedy reverse-rank IM\label{exactgreedyim:alg}}
{\small 
\DontPrintSemicolon
\KwIn{Directed graph $G=(V,E)$, ranker nodes  $Z\subset V$, rankee
  nodes $U\subset Z$, threshold $T$}
\KwOut{Exact greedy sequence $\SeedList$ }
$\SeedList \gets \perp$ \tcp*[h]{Output list of (seed,marginal influence)}\;
\lForAll{rankee nodes $u\in U$}{$\Coverage{u} \gets \emptyset$}
\ForAll{ranker nodes $z\in Z$}
{
$\Coverers{z} \gets \emptyset$\;
 Run Dijkstra from $z$ in $G$, until  (we can determine that) $\pi_{zu} > T$\\
 \ForEach{rankee $u\in U$ with $\pi_{zu}\leq T$}{$\Coverers{z} \gets \Coverers{z}
   \cup \{u\}$\; $\Coverage{u} \gets \Coverage{u} \cup \{z\}$}
}
\While{There are rankees $u$ with $|\Coverage{u}| > 0$}
{
  $v \gets \arg\max_{u\in U\setminus S} |\Coverage{u}|$\;
  Append $(u,|\Coverage{u}|)$ to $\SeedList$ \;
 \ForEach{$z\in \Coverage{u}$}
 {\ForEach{$v\in \Coverers{z}$}
 {Remove $z$ from $\Coverage{v}$}
 }
Delete $\Coverage{u}$
}
\Return{\SeedList} 
}
\end{algorithm2e}

\subsection{Approximate Greedy IM}
We next obtain a near-linear algorithm using two relaxations.
First, 
the greedy selection, and thus the statistical guarantees we obtain,
are with respect to the relaxed influence function \eqref{rrInfT}
where $\hat{\pi}$ replaces $\pi$:
\begin{equation}\label{infhat:eq}
\widehat{Inf}(S) = |\{  z\in Z \mid \hat{\pi}_{zS} \leq T\}|\ . 
\end{equation}
Second, we do not compute an exact greedy sequence for
$\widehat{Inf}(S)$ but instead use 
an approximate greedy algorithm: At each step, selects a 
node with marginal influence that is approximately 
(within a small relative error) the maximum.

Our design adapts the influence 
 maximization algorithms \skim\ and T-\skim 
 \cite{binaryinfluence:CIKM2014,timedinfluence:2014} which are
 designed for reachability-based \cite{KKT:KDD2003} and distance-based 
\cite{Gomez-RodriguezBS:ICML2011,DSGZ:nips2013,timedinfluence:2014}
influence with  threshold functions. 
\notinproc{\footnote{The extension to general $\alpha$ and $\beta$ can be
  achieved by more complex designs which build on 
the distance based designs from \cite{timedinfluence:2014}, but are outside the scope
 of this paper.}}
We quickly review \skim, which 
remarkably, when all nodes are both rankers and rankees,
computes a full approximate greedy permutation in near linear
time.
To do so efficiently, \skim\
samples nodes not covered by previously selected seeds, and maintains for each candidate seed node
the number of sampled nodes it covers.  
Reachability-based \skim\ performs a pruned reverse graph searches 
from the node to determine the nodes that cover it. 
The distance-based \skim\ performs backward pruned Dijkstra searches.
The 
node that first reaches some number $K$ of samples has approximately maximum
marginal influence and is selected as a seed.   The {\em sample-size parameter} $K$ determines a tradeoff between
computation and accuracy. 
 \skim\  then updates the samples so that they are with
respect to the updated marginal influences with the coverage of the
new seed node removed.  \skim\ also updates the representation of the
residual problem. 
The updates are performed using a respective forward 
(Dijkstra) search from the new seed to reveal all nodes that it 
covers.  When a previously sampled node becomes covered, 
the samples of the nodes covering it are adjusted to reflect their
reduced marginal coverage.
Sampling is then resumed until another node reaches
a sample size of $K$. We repeat the 
process of sampling, selecting a seed, and updating the residual
problem until a desired number of seeds is selected, a desired
coverage is achieved, or all nodes are covered.

 Our Algorithm \ref{im:alg} , reverse-rank \skim\ (\rD-\skim), follows the
\skim\ design, of iterating the 
selection of a new seed node (rankee) via sample 
building and updates.  The reverse-rank problem, however, requires some
critical adaptations.

When sample building, we repeatedly select 
random uncovered ranker nodes $z$.  We then 
run Dijkstra's algorithm from  $z$ but stop the search when the
approximate rank
 $\hat{\pi}_{zu}$ exceeds $T$.  For each visited rankee node $u$,  we increment the
 sample size $\Size{u}$ and also add $u$ to a list $\Index{z}$ (the
 list of nodes where $z$ is included in the sample).
This process stops when the first rankee $u$ reaches $\Size{u}=K$.
The node $u$ then becomes the next seed node.
We then apply our sorted-access reverse-rank single source
computation from $u$, up to rank $T$, to determine the coverage of the new
seed $u$.  We mark all uncovered visited nodes as covered.
 For each newly covered ranker $z$, we scan $\Index{z}$ and decrement
$\Size{v}$ for each $v\in \Index{z}$.  We then delete $\Index{z}$.
For each covered ranker $z$, we 
maintain the (approximate) rank of the best seed $\Wasdelta[z].rank =
\min_{v\in S} \hat{\pi}_{vz}$ and the corresponding minimum distance
$\Wasdelta[z].dist = \min_{v\in S} d_{vz}$ (note that the node with
minimum distance must have minimum estimated rank).
 The purpose of maintaining $\Wasdelta$  is to enable pruning of
reverse searches.  Pruning is critical for the near-linear computation
bound of the algorithm (without it, we can construct examples were the
bulk of covered nodes is revisited with each new seed, resulting in 
$\Omega(|\SeedList| m)$ computation).

A search from the new seed $u$ is always pruned at $z$ when
$\hat{\pi}_{uz}>T$, but is also pruned when
\begin{eqnarray}
\lefteqn{\hat{\pi}_{uz} >  \Wasdelta[z].rank \, \text{ or
  }} \label{prunecond}\\
&& \hat{\pi}_{uz} =
 \Wasdelta[z].rank \text{ and } d_{uz} \geq \Wasdelta[z].dist\  .\nonumber
\end{eqnarray}

 We now need to show that also with this pruning,
the algorithm maintains the following invariant
\begin{lemma}
After the processing of a new seed node,
all nodes $z$ with 
$\min_{v\in S} \hat{\pi}_{vz} \leq T$ have
$\Wasdelta[z].rank = \min_{v\in S} \hat{\pi}_{zv}$ 
and
$\Wasdelta[z].dist = \min_{v\in S} d_{zv}$ 
and all other nodes have
$\Wasdelta[z].rank = +\infty$.  
\end{lemma}
\begin{proof}
This property clearly holds when pruning only when
$\hat{\pi}_{uz}>T$, since after inserting
a new seed node $u$ our reverse-rank search from $u$ visits all nodes
with  $\hat{\pi}_{vz} \leq T$.

We  establish the claim using induction on added seeds.
The base of the induction is when $S$ is empty and $\Wasdelta[z].rank = +\infty$.
 Assume now that our invariant holds and let $s_2$ be a newly selected 
seed node.  Let $u_1$ be a node on which we prune the 
search from $s_2$.  
From the condition \eqref{prunecond} there exist  a seed node $s_1$
such that $\hat{\pi}_{u_1,s_1} < \hat{\pi}_{u_1,s_2}$
or 
$\hat{\pi}_{u_1,s_1} = \hat{\pi}_{u_1,s_2}$ and 
$d_{u_1,s_1} \leq d_{u_1,s_2}$.
From the definition of our estimators,
$\hat{\pi}_{u_1,s_1} < \hat{\pi}_{u_1,s_2}$ implies 
$d_{u1,s_1} < d_{u_1,s_2}$. 
Combining, we obtain that $d_{u1,s_1} \leq d_{u_1,s_2}$. 

Now assume to the contrary there is a node $u_2$ such that $u_1$ is on 
the shortest path from $u_2$ to $s_2$ and 
$\hat{\pi}_{u_2,s_2}  <  \hat{\pi}_{u_2,s_1}$
or
$\hat{\pi}_{u_2,s_2}  =  \hat{\pi}_{u_2,s_1}$ and
$d_{u_2,s_2}  <  d_{u_2,s_1}$.   We show that this is 
not possible.


 Using the above and triangle inequality we obtain
 $d_{u_2,s_1} \leq d_{u_2,u1}+d_{u_1,s_1} \leq d_{u_2,u_1}+
d_{u_1,s_2} = d_{u_2,s_2}$.   A property of our estimates is that for
any three nodes
$d_{u_2,s_1} \leq d_{u_2,s_2}$ implies $\hat{\pi}_{u_2,s_1} \leq \hat{\pi}_{u_2,s_2}$.
\end{proof}

The analysis of computation and approximation quality uses components from the 
analysis of $T$-\skim \cite{timedinfluence:2014}.  
An important critical component in the analysis is that we 
can ``charge'' edge traversals used for sample building  to increases in sample sizes. 
When there are many non-rankee nodes, we can construct worst-case 
graph where non-rankees are repeatedly traversed without incrementing 
sample counts.  In realistic models, however, and when all nodes are 
rankers or rankees, we would expect such 
popular ranker hub nodes to  be covered quickly by the first few 
selected seeds. 
Another component of the analysis that carries over from $T$-\skim\
is bounding the number of updates to $\Wasdelta[z]$.  The argument 
there critically 
relies on the sample based approximate greedy selection. 
 The approximation quality of the algorithm can only be 
guaranteed probabilistically and with respect to approximate ranks $\hat{\pi}_{uz}$.
To summarize, when we run the algorithm with $K=O(\epsilon^{-2}\log
n)$, and prune sampling searches using the approximate ranks
$\hat{\pi}$,
we obtain the following.
\begin{theorem}
With very high probability,  for all
$s \geq 1$, the influence $\widehat{\Inf}$ of the first $s$ seed nodes is at least
$1-(1-1/s)^s-\epsilon$ times the maximum possible $\widehat{\Inf}$
with $s$ seeds.
When all nodes are both rankers and rankees, 
the algorithm uses  $O(|E| \epsilon^{-1} \log^3 n +
|E|\epsilon^{-2}\log n)$ operations.
\end{theorem}

\subsection{Approximability of the exact problem}   Distance-based
influence maximization is known to be at least as hard as max cover
also in terms of inapproximability, by a seminal result of Feige \cite{feige98}.  Thus, we know that in a sense
Greedy is the best scalable algorithm.
What we can say about
reverse-rank influence maximization with a threshold kernel  $T$ 
is that it is at least as hard as max cover,  
when each element can be a member of at most $T$ sets.
The problem is NP-hard for $T\geq 2$ (by reduction to max vertex cover), but
Feige's  inapproximability result does not apply.  This leaves open
the possibility that some polynomial-time algorithms have better
approximation ratio than Greedy.

When $T=1$, the influence function is simply  the number of reverse nearest
neighbors.  In this case,  the coverage sets of different nodes are disjoint and
influence maximization is trivial:  The greedy permutation which
selects nodes in decreasing order according to number of reverse
nearest neighbors is optimal.

When $T=2$, each node can be covered by at most two other nodes, which
is similar to max vertex cover, which is also NP hard, but has a
polynomial approximation algorithms that achieves a slightly better
approximation ratio than the greedy guarantee of $1-(1-1/s)^s$
\cite{FeigeLangberg:JAlg2001}.
The Linear Programming based algorithm, however, does not scale for large inputs 
and also does not seem to apply for our general case of $T>2$.


\begin{algorithm2e}[!h]
\caption{reverse-rank \skim} \label{im:alg}
\DontPrintSemicolon
{\small 
\KwIn{Directed graph $G=(V,E)$, ranker nodes  $Z\subset V$, rankee
  nodes $U\subset V$, threshold $T$, parameter~$K$}
\KwOut{Approximate greedy sequence $\SeedList$}

\BlankLine

\tcp{Initialization}
\lForAll{nodes $u\in V$}{$\Wasdelta{u}.rank \leftarrow \infty$}
\lForAll{rankee nodes~$v\in U$}{$\Size{v} \leftarrow 0$}
$\Index \leftarrow \perp$\tcp*{ Hash map of ranker nodes to sets of
  rankee nodes}
$coverage \gets 0$ \tcp*{Coverage of current $\SeedList$}
$\SeedList \leftarrow \perp$\tcp*{Output list}
$F\gets$ random shuffle of the ranker nodes $Z$\;
\BlankLine
\While(\tcp*[f]{select seed}){$|coverage|<|Z|$ \And  ($\exists$ unscanned $u\in F$ \Or
  $\max_{u \in U}{\Size{u}}> 0$)}{
$x \gets \perp$ \tcp*{next seed node}
\While(\tcp*[f]{Build samples}){$\exists$ unscanned $u\in F$}{
	$u\leftarrow$ next node in shuffled sequence $F$\;
	\BlankLine
    \If(\tcp*[f]{Node $u$ is covered, skip it}){$\Wasdelta{u}.rank <
      \infty$}{\Continue} 
\tcp{Find all rankees $v$ with $\hat{\pi}_{uv}\leq T$}
    Run a Dijkstra search from~$u$ in~$G$, during which\;
    \ForEach{scanned rankee node $v\in U$}{
    	\lIf{$\hat{\pi}_{uv} > T$}{terminate the search}
      $\Size{v} \leftarrow \Size{v} + 1$\;
      $\Index{u} \leftarrow \Index{u} \cup \{v\}$\; \BlankLine
      \If{$\Size{v} = K$}{
        $x \leftarrow v$\tcp*{Next seed node}
        abort sample building loop\;
      }
    }
}

\BlankLine
\If{$x = \perp$}
{$x \leftarrow \arg\max_{u \in U}{\Size{u}}$\;
\lIf{$\Size{x}=0$}{abort main loop}}

\BlankLine
$I_x \leftarrow 0$\tcp*{Estimated coverage of~$x$}

\tcp{Compute $I_x$ and update residual}
	run pruned reverse-rank single-source search from~$x$ in transposed graph~$G^T$, during which\;
	\ForEach{scanned node~$v\in V$ with ($\hat{\pi}_{vu}$,
          $d_{vu}$) }{
		\If{$\hat{\pi}_{vu}  > T$ \Or $\Wasdelta{v}.rank <
                  \hat{\pi}_{vu} $ \Or $\Wasdelta{v}.rank =
                  \hat{\pi}_{vu}$ \And $\Wasdelta{v}.dist \leq
                  d_{vu}$}{\Prune at $v$}
		\If(\tcp*[f]{$v$ is a newly covered ranker}){$\Wasdelta{v} = \infty$ \And $v\in Z$}{$I_x \leftarrow I_x + 1$ ;
                $coverage \gets coverage+1$\;
		\BlankLine 
		\ForAll{nodes $w$ in $\Index{v}$}{
			$\Size{w} \leftarrow \Size{w} - 1$\;
		}
		$\Index{v} \leftarrow \bot$ \tcp*{Delete}
}
		$\Wasdelta{v}.rank \leftarrow \hat{\pi}_{vu}$\;
		$\Wasdelta{v}.dist \leftarrow d_{vu}$\;

	}
\BlankLine
\SeedList.\Append{$x$, $I_x $}\;
}
\Return{\SeedList}\;
}
\end{algorithm2e}

\section{Hardness of exact reverse-rank single source
  computation} \label{hardness:sec}

  Exact single source reverse-rank computation from a node $u$ will return
  $\pi_{iu}$ for all nodes $i$.   Clearly,  it can be solved using an
  APSP computation.  We show the following
\begin{theorem} \label{hardness:thm}
The reverse-rank single source problem has subcubic equivalence to APSP.
\end{theorem}

 We give a  reduction from
  the
{\em Graph Radius Problem}:
The radius of a graph $G$, is defined as the minimum over nodes $u$ of
the maximum distance from $v$ to another node
\begin{equation} \label{radius:eq}
R = \min_{u \in V} \max_{v \in V} d_{uv}\ . 
\end{equation}
The graph radius problem on undirected  graphs
is known to have subcubic equivalence to
APSP \cite{AGV:SODA2015}.  

Given a graph $G=(V,E)$, and a length parameter $x$, 
we construct a new graph $G_x=(V',E')$ by
adding a new node $V' = V\ \cup \{u\}$ and adding edges from $u$ to
all $v\in V$ with length $x$. 
\begin{lemma}
Let $G$ be a graph with radius $R$ and consider $x>0$.
If $R> x$ then in $G_x$,   for all nodes $v\in V$,
$\pi_{vu} < |V|$.   If $R<x$ then there must exist a node $z\in V$ such
that $\pi_{zu}= |V|$ in $G_x$.
\end{lemma}
\begin{proof}
Suppose that $R>x$ then 
by definition of radius, for all nodes $z$,
$\max_{v \in V} d_{zv}  \geq R > x$.  Therefore node $u$ will not be
the farthest from $z$ and we have 
$\pi_{zu} < |V|$.

Suppose now that $R<x$. 
let  $z\in V$ be such that $R=\max_{v \in V} d_{zv}$.
 Then for all $v\in V$,  $d_{zu}=x > R \geq d_{zv}$ and thus $\pi_{zu}
 = |V|$.
\end{proof}

  From the lemma, we can compute the graph radius $R$ by performing a
  logarithmic number (in the representation of $G$) of exact reverse-rank single source computations on
  graphs the size of $G$.  This concludes the proof of Theorem \ref{hardness:thm}.

\begin{corollary} \label{hardnessinfluence:coro}
Exact computation of reverse-rank influence, even with a single seed
node $|S|=1$, uniform $\beta$, and a threshold function $\alpha$, is
sub-cubic equivalent to APSP.
\end{corollary}
\begin{proof}
We use the same construction and compute influence (centrality) of
node $u$ with $\alpha$ being a threshold function with $T=|V|-1$.
\end{proof}

\ignore{
\section{Score aggregation}

  We can select a small set $S$ of important or  high 
  influence seed nodes.  We can compute reverse-ranks for these beacons,
  obtaining for all nodes in the graph how they rank the beacon nodes.
The rank with respect to the beacons can be used to ``position'' a
node in the network.  Similarly,  we can evaluate similarity of two
nodes based on the similarity of how they rank the beacons.
}

\ignore{
\section{\rD -rank oracle} \label{oracle:sec}

  ADS sketches also can be used to obtain upper and lower bounds on
  distances \cite{CDFGGW:COSN2013}.  When the graph is undirected,
  they obtain probabilistic approximation guarantees close to the
  worst-case best.  In practice, approximation is much tighter.
  Moreover, a better approximation corresponds to more short paths between
  the nodes, which correlates well with a tighter relation we expect in
  social graphs.

  We can obtain bounds on the \rD -rank of $i$ with respect to $j$, $\pi_{ij}$ by estimating the
  distance $d_{ij}$.  Upper (lower) bounds on the distance provide us with
  upper (lower) bounds on the rank. 

  Moreover, having the ADS and the distribution, we can also work with
  ranking that is robust to small perturbations in weights.  We can
  take the rank as the median rank of nodes that have the same
  approximate distance as $j$ from $i$.
}

\section{Experimental evaluation} \label{experiments:sec}

  We implemented and evaluated our algorithms for computing ADS,  approximate reverse-rank
  single-source,  and influence maximization.
  Our implementations are  in C++ and were compiled using gcc (g++) with full
  optimization.  Our testing machine runs Centos 6.5 and uses 
Dell PowerEdge R720 server with two Intel Xeon E5-2640 CPUs.  Each
with 12 cores  (2.50GHz, 12$\times$32kiB
L1,  6$\times$256kib L2, and 15MiB L3) 
and 264GiB of RAM. The disk capacity is 1T.

  Table~\ref{ADStable} shows the social graphs used for our evaluation, all taken from the SNAP
  project \cite{SNAP}. For each graph we list the number of nodes
  and edges and whether edges are directed.  These data sets did not
  distinguish between edges, so we used uniform lengths of 1.  Our
  implementations, however, are designed to work with general positive
  edge lengths. 

\subsection{Sequential ADS computation}
  Table~\ref{ADStable} also lists, for each instance, performance
  figures (time and memory usage) of our optimized
 sequential implementation of {\em Pruned Dijkstras} (Algorithm
 \ref{pDijkstra:alg}).  
We list performance for ADS parameter values
  $k=16,64,128$  (higher $k$ implies larger sketch size and
  processing and higher estimation quality).  The listed times are
  broken into {\em load time} -- loading the graph into  memory data
  structures, {\em ADS time} -- computing the sketches,
  and {\em ests time}  -- process the ADS sketches to compute the
  distance to cardinality estimation lists $L(i)$.
We can see that ADS computation was the dominant component.
Overall, the preprocessing time is of the order of few hours, even on our largest data
set. The table also lists the virtual memory usage of the different runs.
For reference, we provide in Table~\ref{KNN:table} the running time of computing the $T$
nearest neighbors for all nodes in the graph for $T=16,64,128$.  We
can see that our ADS computation times are comparable with this
simpler operation.

\begin{table*}[!t]
\caption{Test Instances and preprocessing  time (single thread) \label{ADStable}}
\centering
{\footnotesize
\begin{tabular}{lrrrrrrrrrrrr}
\hline
 &  &  &  &
                                            \multicolumn{3}{c}{\underline{Preprocess 
            ($k=16$)}}  & \multicolumn{3}{c}{\underline{Preprocess
                         ($k=64$)}
} &
                                                                  \multicolumn{3}{c}{\underline{$\quad$
    Preprocess
                                                                  ($k=128$)
    }}
  \\
 instance &  &  & load &
 ADS & mem&  ests &
 ADS & mem&  ests &
 ADS & mem&  ests \\
(un/directed) & \#nodes & \#edges  & [Sec] 
& [Sec] &[GB] &  [Sec] &
 [Sec] & [GB] &  [Sec] &
 [Sec] & [GB] &  [Sec] \\
\hline
\instance{Facebook} (u)	 & 3,959	 & 84,243	&  0.3	&  0.5	&  0.03	&  0.043 
&  0.7	 & 0.03	&  0.140	&  1.4	&  0.04	&  0.341 \\
 \instance{Slashdot} (d)	 & 77,360	&  905,468	&  0.5	&  5.7	&
 0.1	&  0.276	&  23 & 	 0.22 & 	 1.91 & 	 42 
 &  0.4	&  5.35\\
 \instance{Twitter} (d)	&  456,626	&  14,855,842	&  22 & 	 88 &
 0.9 & 	 2.00	&  338	&  1.4	&  13.5	&  523	&  2 
 &  39 \\
 \instance{YouTube} (u)	&  1,134,890	&  2,987,624	&  3.3 & 	 116 & 	 1.4 &
 6.69	&  404	&  3.3	&  41.6	&  770 & 	 5.3	&
 118 \\
\instance{LiveJournal} (u)&  3,997,962	&  34,681,189	&  32	&  1481	&  6.4 
 & 29.7	&  4,901	&  18	&  209  & 	 9,555 & 	 31.5	&
 642\\
\hline
\end{tabular}
}
\end{table*}
\ignore{
\instance{Facebook} (u)	 & 3,959	 & 84,243	&  0.2	&  0.3	&  0.03	&  25 
&  0.9	 & 0.03	&  122	&  1.7	&  0.04	&  362 \\
 \instance{Slashdot} (d)	 & 77,360	&  905,468	&  0.5	&  0.9	&
 0.1	&  364	&  29.7 & 	 0.22 & 	 1,930 & 	 58.9 
 &  0.4	&  5,558\\
 \instance{Twitter} (d)	&  456,626	&  14,855,842	&  21 & 	 191 &
 0.9 & 	 4,514	&  414	&  1.4	&  15,444	&  776	&  2 
 &  40,409 \\
 \instance{YouTube} (u)	&  1,134,890	&  2,987,624	&  3 & 	 157 & 	 1.4 &
 10,530	&  450	&  3.3	&  43,093	&  945 & 	 5.3	&
 120,726 \\
\instance{LiveJournal} (u)&  3,997,962	&  34,681,189	&  30	&  1959	&  6.4 
 & 52,015	&  5,425	&  18	&  213,331 & 	 10,437 & 	 31.5	&
 635,173\\
}

\subsection{Multithreaded ADS algorithm}
 We next evaluate our implementation of the multithreaded ADS 
algorithm (Section \ref{multithread:sec}). 
The evaluation was done by generating  1 to 14 concurrent threads. 
We used batch size
parameter $\mu=0.1$.   The parameter $\mu=0.1$ was selected since it
had the best performance on a test of sweeping $\mu$ between $0.05$
and $1$ and considering 1--14  threads on the \instance{slashdot}
graph.
We note that the amount of concurrency provided in the algorithmic design is 
much larger, but the architecture of our machine, mainly number of
cores and shared caches, limited the benefit of using more threads.
We show the results for executions with ADS parameter $k=16$.
 The time to load the graph into memory and the total virtual memory 
 used did not vary much for the same instance and different numbers of 
 threads.  Table \ref{multi:table} lists the load-time and virtual 
 memory numbers for 7 threads. 
The table also shows the run time on a single thread. Note that it can be
slightly larger  than our optimized sequential implementation.  Figure 
\ref{multithread:fig} shows the running times, as a fraction of the
running time on a single-thread,  as a function of the number of 
concurrent threads.  We note that the number of threads listed is the
concurrency generated by our program scheduler --
the actual number of cores allocated by the OS was sometimes smaller 
and we had no control over it. 
We observe significant benefit of the multithreaded design, in
particular for the larger graphs where we obtain up to a factor of
$3$ speedup, also with respect to the optimized sequential implementation.  
  We note that most of the speedup is obtained using only
$2-6$ threads.

\begin{table}[H]
\centering 
{\small
\caption{Multithreaded ADS computation\label{multi:table}}
\begin{tabular}{lrrr}
   & load  & Memory & 1-thread \\
instance &  [Sec] &  [GiB] & [Sec]\\
\hline 
\instance{Facebook} & 0.13 & 0.56  & 0.69 \\
 \instance{Slashdot} & 0.57 & 1.1 & 10.0  \\
\instance{Twitter} & 23 & 2.2 & 120 \\
\instance{YouTube} & 3.9 & 3.3 & 157 \\
\instance{LiveJournal} & 36 & 11 & 1541 
\end{tabular}
}
\end{table}
\begin{figure}
\centering 
\includegraphics[width=0.45\textwidth]{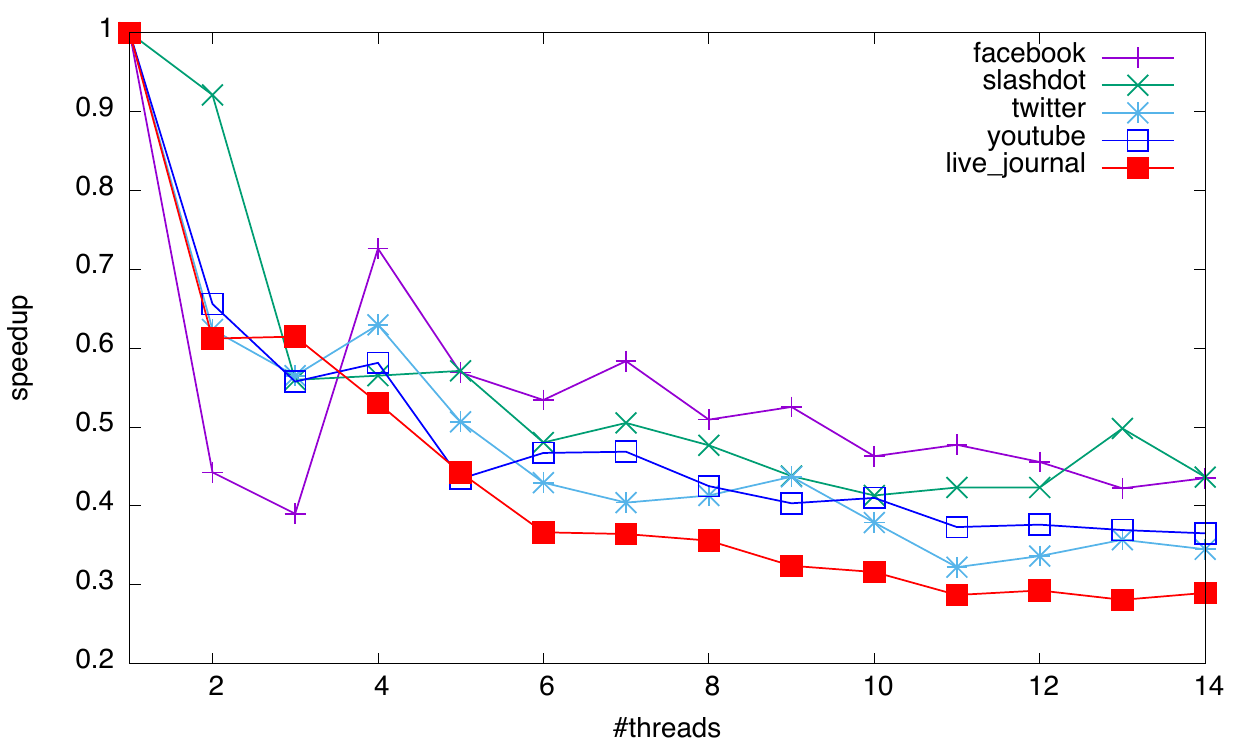}%
\caption{\label{multithread:fig}{Multithreaded computation. Speedup 
    (ratio of time to single-thread time) as a function of the number 
    of threads.}}
\end{figure}

\subsection{Reverse-rank single-source computation}
Table~\ref{SingleSourcetable}  shows running times of our approximate
reverse-rank single source computations, averaged over 1000 source nodes
selected uniformly at random.  The times listed are net per
computation after loading the graph and pre-computed sketches $L(i)$
into memory. We show running times for the different ADS
parameter values $k=16,64,128$.  For reference, we  also show
running time for Dijkstra's algorithm (single-source distances)
averaged over the same 1000 source nodes.  We can observe that the
running times of our reverse-rank single-source computation do not
depend on the sketch parameter $k$ and are similar to Dijkstra
computations.
The table also shows extrapolated running 
time for APSP computation.  The extrapolation was obtained by multiplying the time 
for a single Dijkstra run by the number of nodes.  This is listed for 
reference, since exact  reverse-rank single-source computation is 
equivalent to an all-pairs computation.

The table also displays the average relative errors (ARE) for each
sketch parameter.   Since it was not possible to scalably compute the
exact reverse-rank values even for a single source,  we computed instead the
estimation errors on the ranks using the Dijkstra runs:   The errors were therefore
averaged over all the ranks provided by 1000  different rankers
instead of all the ranks ``received'' by 1000 different rankees.
We can see that the ARE, as 
expected, decreases with the ADS parameter $k$ and are within
the theoretical bounds.  Note that a fixed set of
sketches was computed during preprocessing and used in  all subsequent
computations.  Therefore the estimates on reverse-ranks of different source-destination pairs
can be highly dependent.

{\small
\begin{table*}
\caption{Reverse-Rank Single Source computation, averaged over 1000 source nodes
  selected uniformly at random. \label{SingleSourcetable}}
\centering
\begin{tabular}{lcccccccc}
\hline
& \multicolumn{6}{c}{\underline{$\quad$ reverse-rank single source
    $\quad$}} & Dijkstra  & APSP \\
& \multicolumn{2}{c}{\underline{$\quad k=16 \quad$}} &
                                                       \multicolumn{2}{c}{\underline{$\quad
                                                       k=64 \quad$}} &
                                                                       \multicolumn{2}{c}{\underline{$\quad
                                                                       k=128
                                                                       \quad$}}
& & \\
instance & [Sec] & ARE & [Sec] & ARE & [Sec] & ARE & [Sec] &  $\approx$ [hours] \\
\hline
\instance{Facebook} & 0.006 &  0.11 & 0.006 & 0.072 & 0.006& 0.067 & 0.002  &	 0.002\\
 \instance{Slashdot} &0.071 & 0.34  &	 0.077	& 0.12 &  0.074 & 0.076
 &  0.06 & 1.3 \\
\instance{Twitter} &0.96 & 0.29	& 1.00 & 0.050	& 0.88 & 0.076 & 0.54 & 68 \\
\instance{YouTube} & 1.44 & 0.24 &	 1.51	&0.041 & 1.34 & 0.023 & 0.98& 309 \\
\instance{LiveJournal} & 13 &  0.11 & 16 & 0.13 	& 16 & 0.053	& 5.8 & 6440\\
\hline
\end{tabular}
\end{table*}
}
\ignore{
\instance{Facebook} & 0.004 &  & 0.004 & & 0.005& & 0.0033  &	 \\
 \instance{Slashdot} &0.12 & &	 0.1	& &  0.1 & &  0.0805 & \\
\instance{Twitter} &1.1& 	& 1.1& 	& 1.3 & & 0.78 & \\
\instance{YouTube} & 2.1 & &	 2.1	& & 2.2 & & 1.33& \\
\instance{LiveJournal} & 14.2 & & 15.2& 	& 17.8 & 	& 7.891 & \\

}


\subsection{Reverse-rank distributions}
  Our implementation allowed us, for the first time, to view the
  reverse rank distributions of  nodes in a large network.
Figure \ref{reverseranks:fig} (left) shows the cummulative reverse rank
distributions $\overline{\pi}_{js}$ for 4 selected source nodes in the
\instance{YouTube} network.  For each node $s$, we sort the
(estimated) values $\overline{\pi}_{js}$ for all nodes $j$  in
increasing order.  The cummulative distribution plot then shows the
value $y$ at each position $x$.
The figure also includes a 
reference line where for any $i$ there 
is a node with rank $i$.  The reference line is in a sense corresponds
to an  ``average''  source node, which gives and receives the same influence.

We can get information on the relative importance of a node in its
``locality,''  for varying locality ranges,  from its reverse-rank
distribution.
Nodes that are important in their locality 
 would have distributions that at least initially lie well 
below the reference line.  
This means that for some $i$,
there are many more than $i$ nodes that rank them below $i$.
Node \#2711 and \#480 are example influential nodes that remain
important across neighborhood scales.
Node \#368749 has low influence with distribution  above the
reference line across ranges.
 Node \#3394 has low influence on most ranges except for its immediate
 neighborhood, where it has average influence,  and on the longest
 scale, when looking at its 
$7\times 10^5$ and above highest rankers (which is 35\% of total nodes), which indicates that
it lies closer to the ``core'' of the network.
Note that we plot $\overline{\pi}$, meaning ties are broken
``upwards,'' which biases towards being above the reference line.

Figure~\ref{distdist:fig} (right)  provides, for comparison,  the cummulative distance
distributions for the same nodes: For each number of hops $y$, we see the
number $x$ of nodes within $y$ hops.   
The distance distribution captures the expansion rate,  but does not quantify well the
relative status of a node within its locality:
A less influential  member of a dense community would have higher
expansion than an influential member of a sparser community.
As a simplified example think of  two nodes $A$ and $B$ with the same
degree $\Delta$ such that all neighbors of $A$ have degree $\ll
\Delta$ and all neighbors of $B$ have degree $\gg \Delta$.  In this
case we may view $A$ as being influential in its neighborhood whereas
$B$ will not be.  The reverse-rank distribution will correctly make
this distinction whereas the distance-distribution will not.

\notinproc{
\begin{figure*}
\centering 
\includegraphics[width=0.45\textwidth]{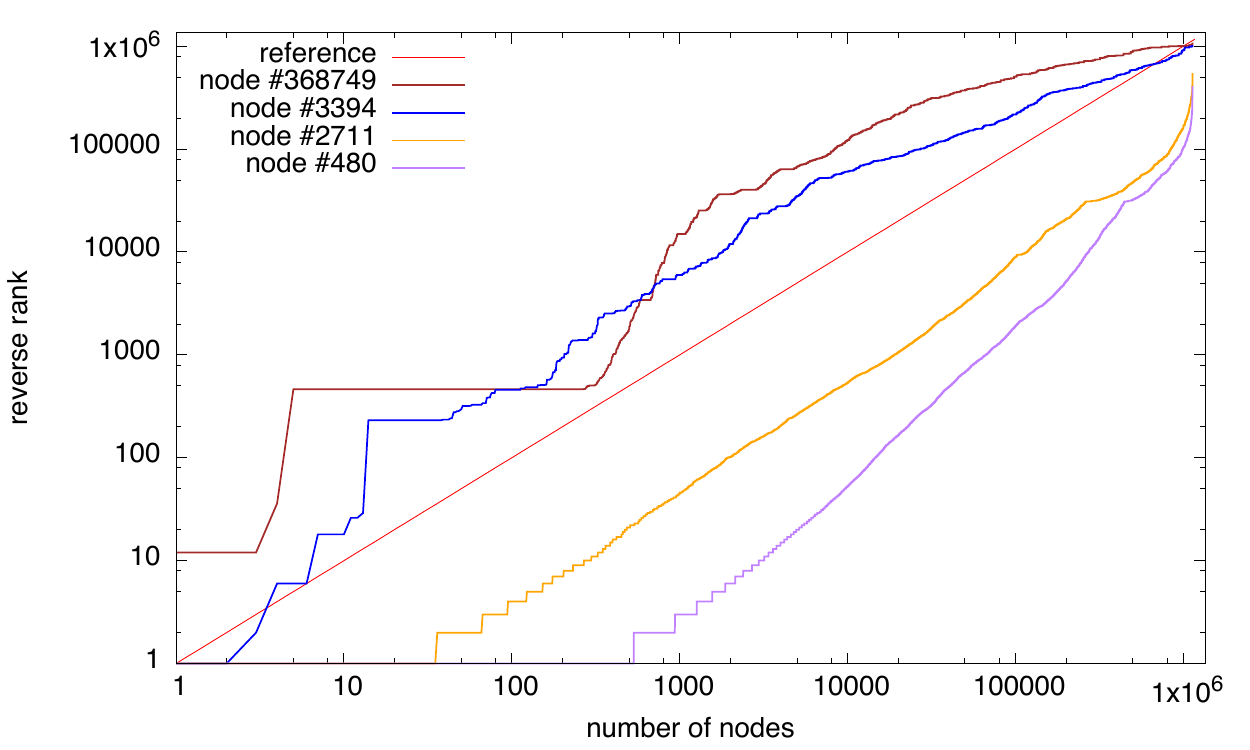}%
\includegraphics[width=0.45\textwidth]{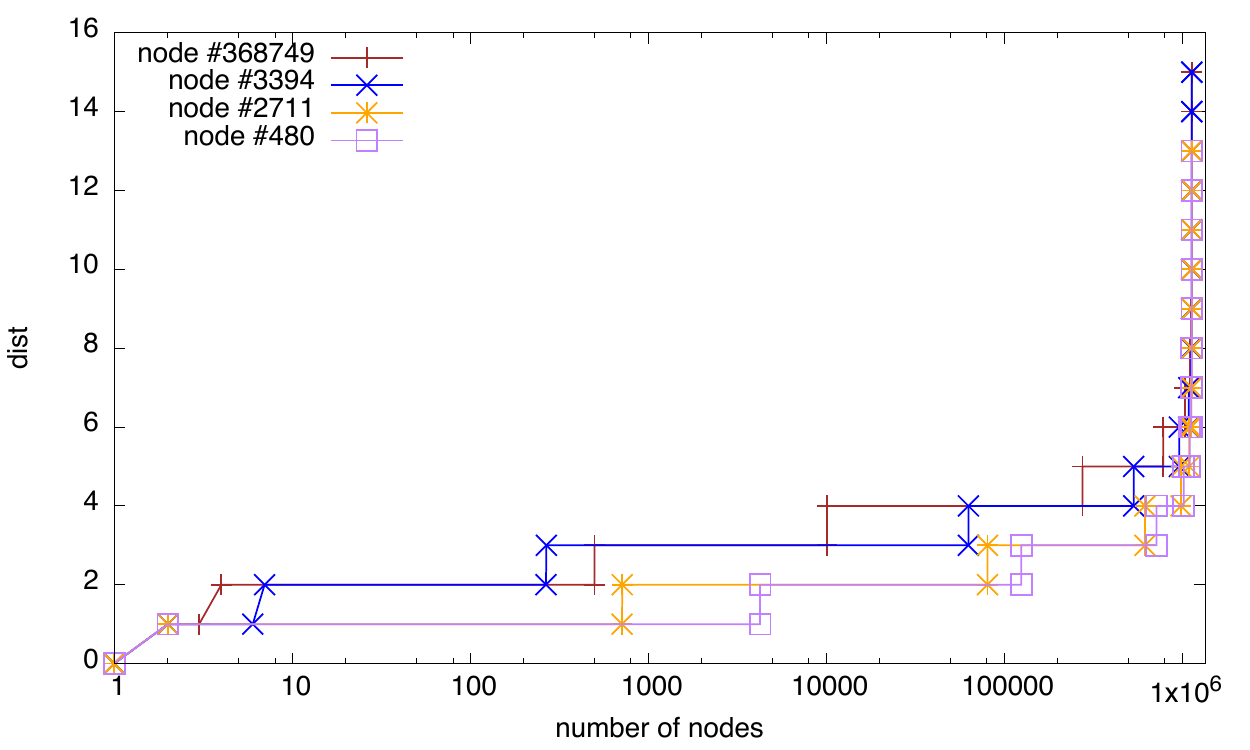}%
\caption{\label{reverseranks:fig}\label{distdist:fig}{Cummulative distributions on \instance{YouTube} graph. Left:  Reverse-rank Right: Distance.}}
\end{figure*}
  }

\onlyinproc{
\begin{figure}
\centering 
\includegraphics[width=0.22\textwidth]{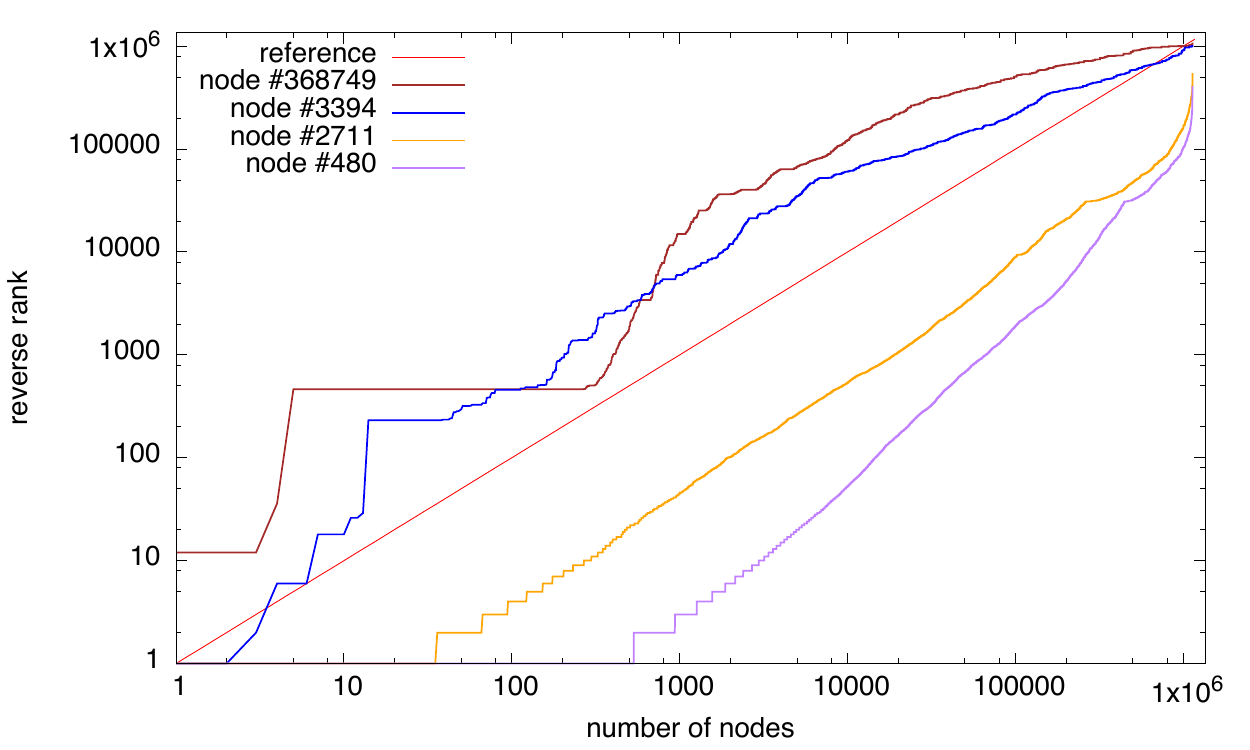}%
\includegraphics[width=0.22\textwidth]{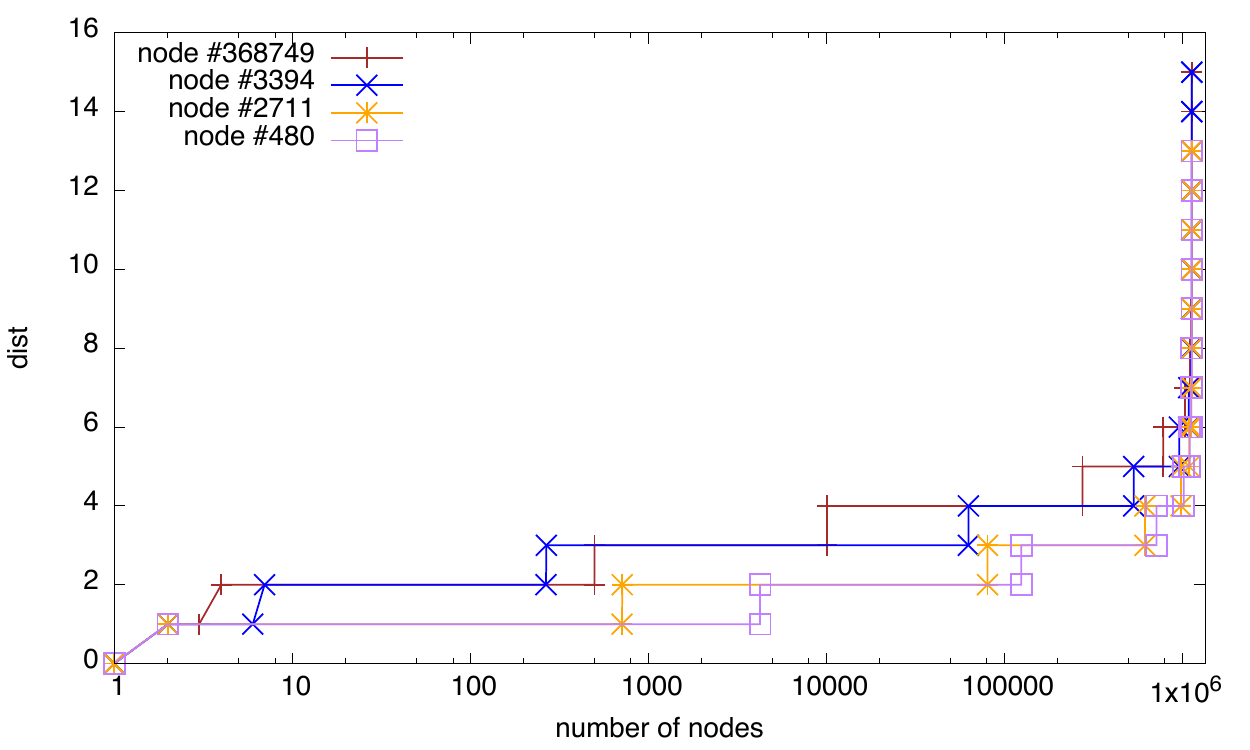}%
\caption{\label{reverseranks:fig}\label{distdist:fig}{Cummulative distributions on \instance{YouTube} graph. Left:  Reverse-rank Right: Distance.}}
\end{figure}
  }

\subsection{Influence maximization}
\notinproc{
\begin{figure*}
\centering 
\begin{tabular}{c}{\includegraphics[width=0.30\textwidth]{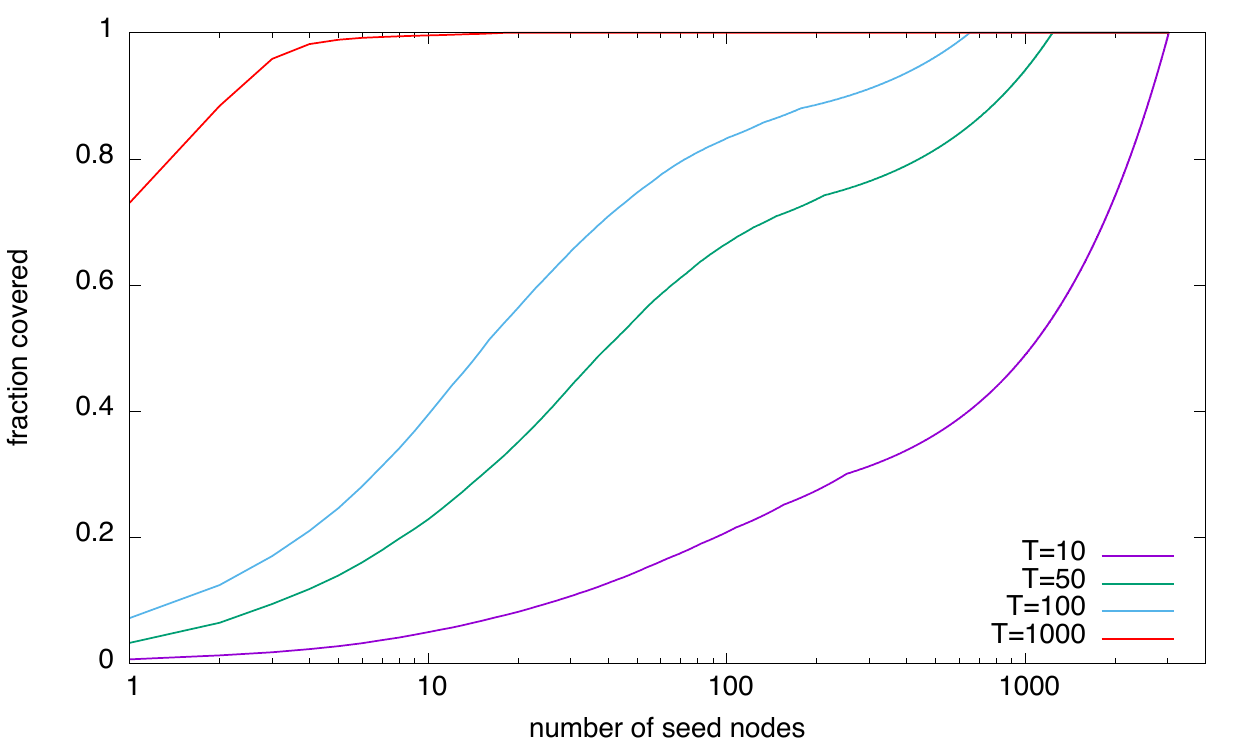}}\\{\small
  \instance{Facebook}}\end{tabular}%
\begin{tabular}{c}\includegraphics[width=0.30\textwidth]{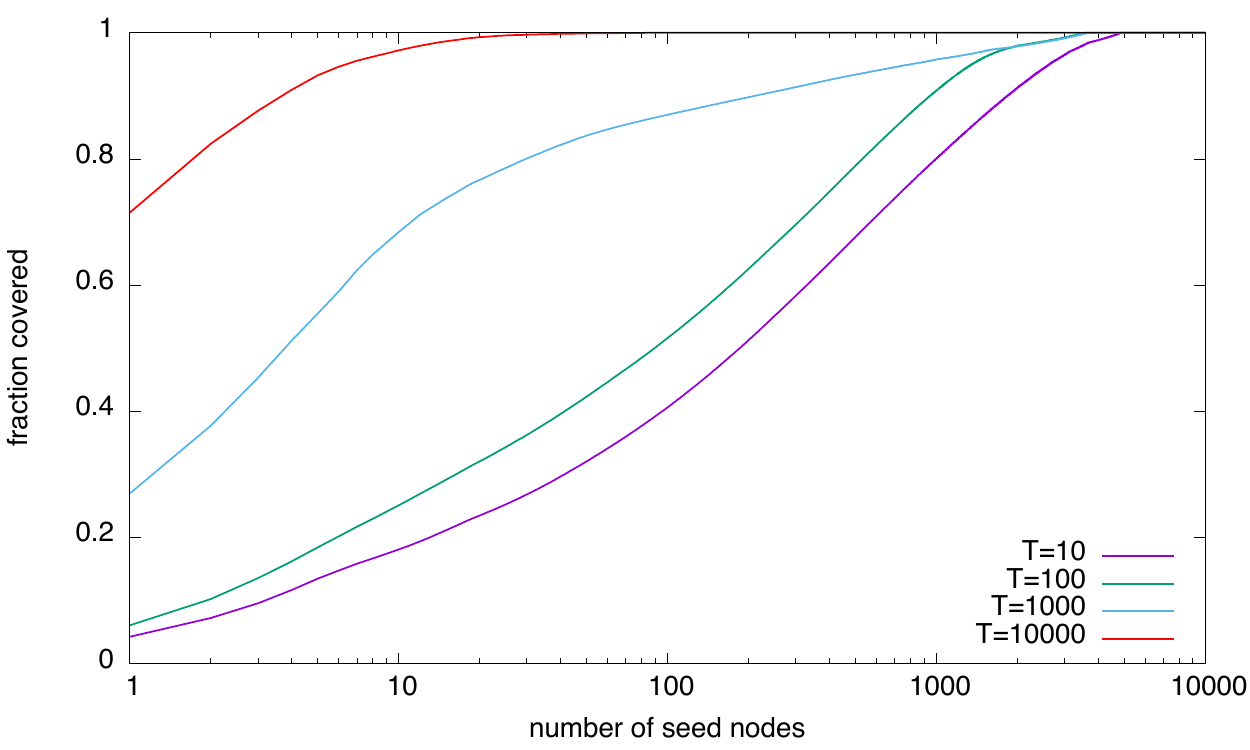}\\{\small
  \instance{Slashdot}}\end{tabular}%
\begin{tabular}{c}\includegraphics[width=0.30\textwidth]{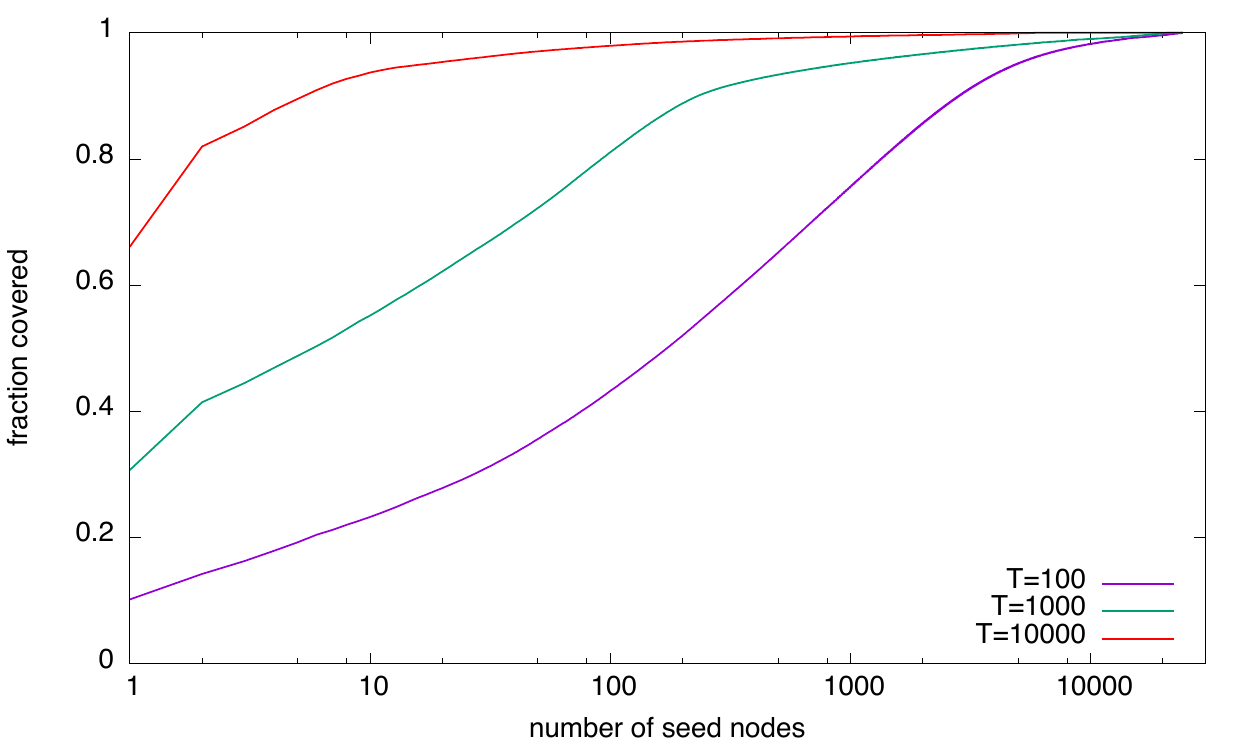} \\{\small
  \instance{Youtube}}\end{tabular}\\
 \begin{tabular}{c}
 \includegraphics[width=0.30\textwidth]{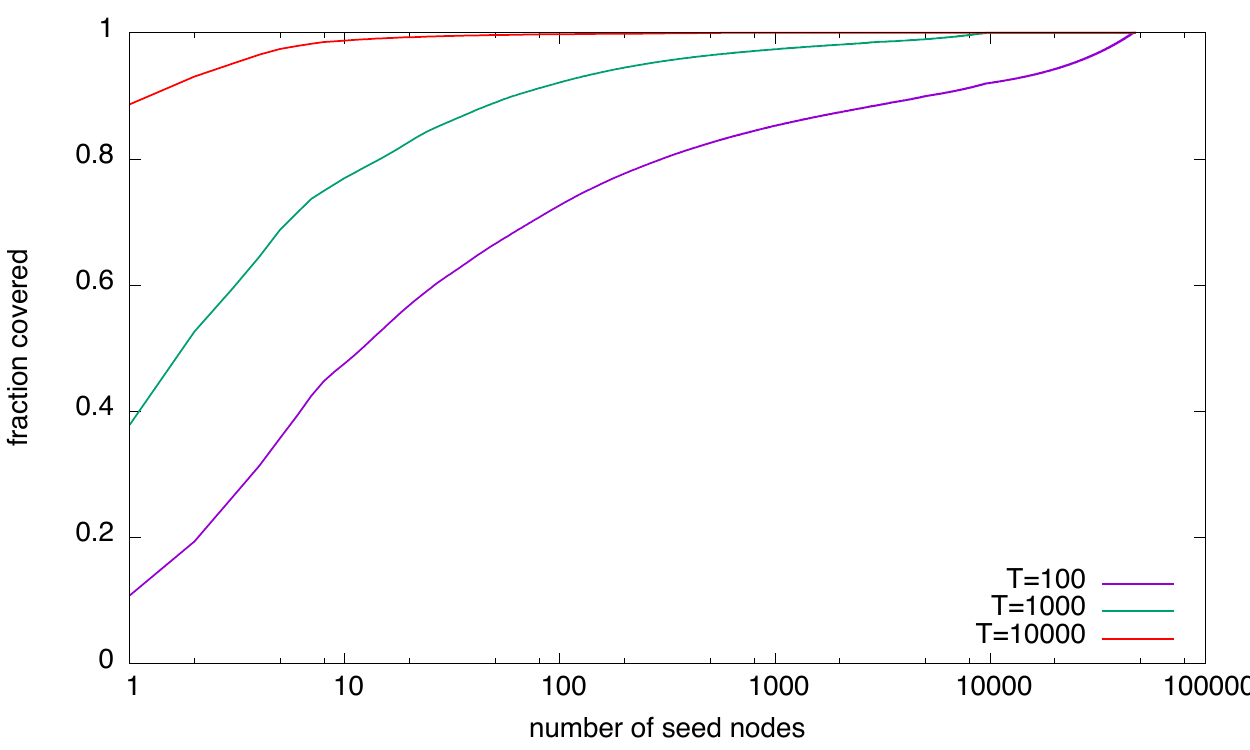}\\
  {\small    \instance{Twitter}}
  \end{tabular}%
 \begin{tabular}{c}\includegraphics[width=0.30\textwidth]{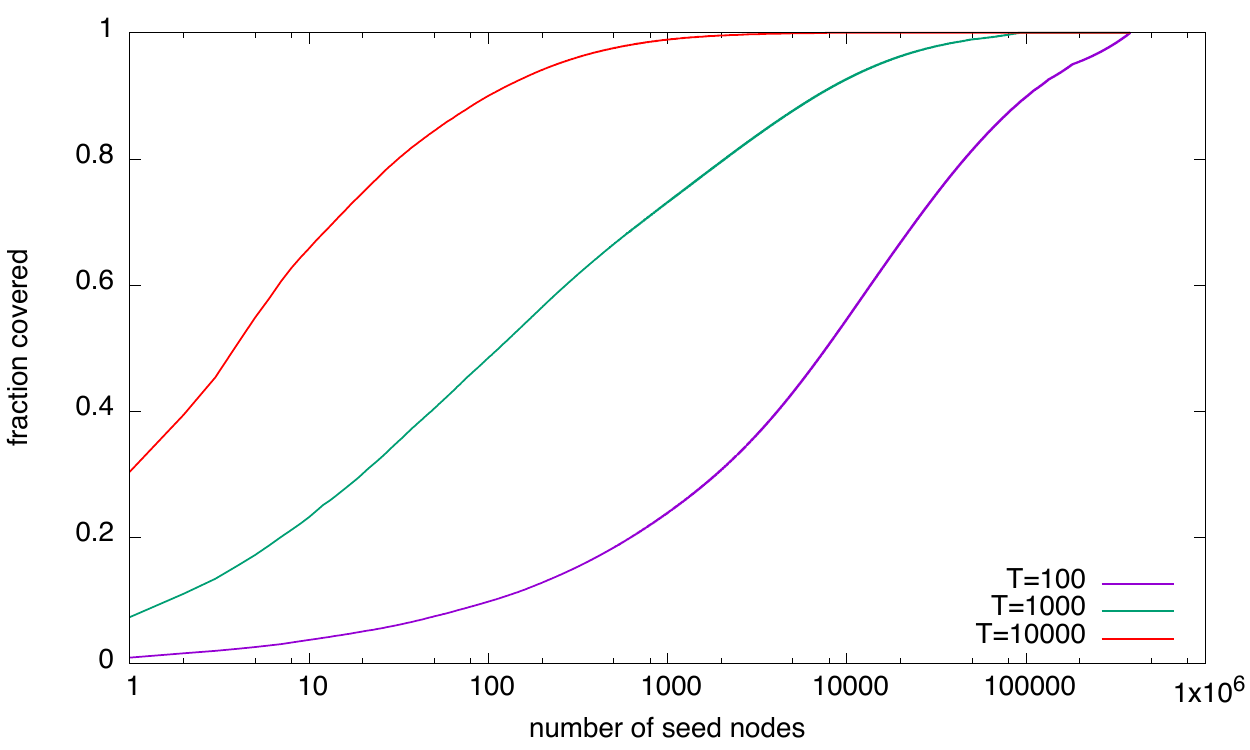}\\{\small
   \instance{LiveJournal}}\end{tabular}%
\caption{\label{coverageT:fig}}{Fractional coverage $\widehat{\Inf(S)}/|V|$ as a
  function of seed set size $|S|$ for the greedy sequence, for different values of $T$.}
\end{figure*}
}
\onlyinproc{
\begin{figure}[t]
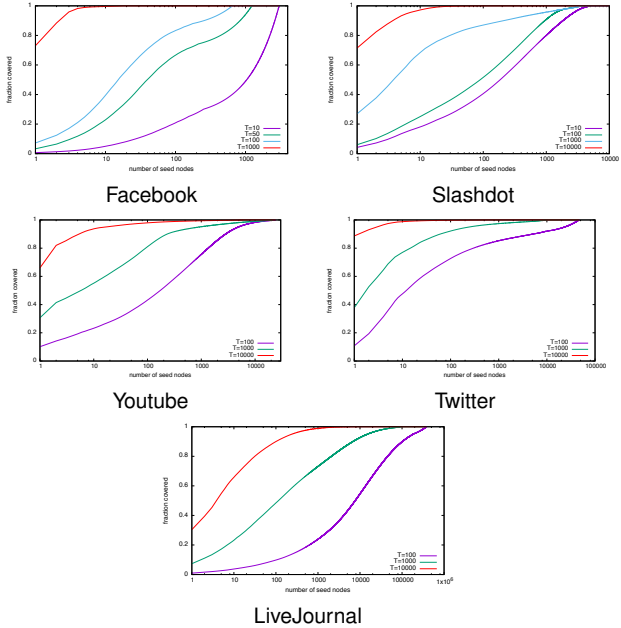

\centering 
\begin{tabular}{c}{\includegraphics[width=0.22\textwidth]{coverageT-facebook.pdf}}\\{\small
  \instance{Facebook}}\end{tabular}%
\begin{tabular}{c}\includegraphics[width=0.22\textwidth]{coverageT-slashdot.pdf}\\{\small
  \instance{Slashdot}}\end{tabular} \\
\begin{tabular}{c}\includegraphics[width=0.21\textwidth]{coverageT-youtube.pdf} \\{\small
  \instance{Youtube}}\end{tabular}
 \begin{tabular}{c}
 \includegraphics[width=0.21\textwidth]{coverageT-twitter.pdf}\\
  {\small    \instance{Twitter}}
  \end{tabular}
 \begin{tabular}{c}\includegraphics[width=0.22\textwidth]{coverageT-livejournal.pdf}\\{\small
   \instance{LiveJournal}}\end{tabular}%
\caption{\label{coverageT:fig} Fractional coverage $\widehat{\Inf(S)}/|V|$ as a
  function of seed set size $|S|$ for $\rD$-\skim, for varying $T$.}
\end{figure}
}
  We next evaluate the performance of reverse-rank \skim, in terms of
both the quality of the coverage and the running time.  The evaluation used the
social graphs listed in Table~\ref{ADStable}, with uniform edge
lengths and  all nodes being both rankers and rankees.
We used the $\pi\equiv 
\overline{\pi}$ \eqref{ntopi} interpretation of rank.
  We study dependence on
three parameters, $T$, $k$, and $K$:  
The threshold value $T$ specifies the
coverage rate.
The ADS sketch  parameter $k$
determines the quality of 
$\hat{\pi}$ as estimates of the true ranks $\pi$, and thus, the relation
between $\widehat{\Inf}$ \eqref{infhat:eq}, which we optimize for, and the true $\Inf$.
Finally, the sample size parameter $K$
determines the quality of the coverage
 in terms of the approximate influence   $\widehat{\Inf}$
\eqref{infhat:eq}.  Larger $K$ mean that we are more likely to
select seeds with marginal $\widehat{\Inf}$ influence that is closer
to the maximum.  

Recall that the computation of an exact greedy sequence, with
respect to either $\widehat{\Inf}$ or $\Inf$ (Algorithm~\ref{exactgreedyim:alg}), is $\tilde{O}(T |E|)$,
where the $\tilde{O}$ notation suppresses logarithmic factors.  
The computation of reverse-rank \skim\ uses ADS computation of
$\tilde{O}(k |E|)$ and additional computation with worst-case bound of 
$\tilde{O}(K |E|)$. Moreover, note that in actuality, the time is
$\tilde{O}(K|E|\rho)$, where $\rho$ is the ratio between the
average and maximum influence of a node, and for typical skewed
influence distribution, we have $\rho\ll 1$.   Therefore, we expect the
scalability advantage of \skim\ to become more significant for larger $T$.

 When evaluating the effect of the sample size $K$,
we fixed the  ADS 
 parameter to be $k=128$ for the four smaller graphs and $k=64$ for the 
 larger one and used  several values of $T \leq 10^4$.
We computed the exact greedy selection with respect to
$\widehat{\Inf}$, which is obtained by selecting a node with maximum
marginal $\widehat{\Inf}$ in each step.  This was done on 
the four smaller graphs.  On these graphs, sample size $K=100$ was almost
   always within a fraction of a percent of the exact greedy
   coverage.  The one exception was on 
 \instance{Slashdot} and $T=10^4$ where the first seed had coverage
 that is 6\% lower than the optimal one and the gap closed up for the
 first three seeds.
With $K=500$, the approximate greedy coverage almost exactly
   matched the exact greedy coverage.  For \instance{LiveJournal}, we
   only evaluated the coverage for sample size up to $K=500$, but the
   performance with $K=100$ already matched that.  We note that these
   observed errors are much lower than the worst-case guarantees
   provided in our analysis.  The explanation is the
skew of the influence distribution is skewed, where
   the node of maximum marginal influence is well separated from the
   second maximum, and with very few nodes having influence that is
   more than a fraction of the maximum.

  Our implementation allows us to examine the coverage to seed set size
  tradeoffs as a function of  the threshold $T$.  These tradeoffs
  provide structural insights on the networks and
results are shown in Figure~\ref{coverageT:fig}.  
 Higher values   of $T$ as expected have higher coverage with fewer
 seeds.  We can also see a highly skewed and asymmetric distribution
 of importance.  For example, the \instance{LiveJournal} graph with
 nearly four million nodes, there is a single node that 4$\times 10^4$
 other nodes rank within their top $T=100$.  The first 11 nodes have
 $1.6 \times 10^5$ nodes ranking at least one of them in their top
 $100$.  For $T=1000$, the top seed covers $3 \times 10^5$ nodes  and
 the top 12 cover $7.5\times 10^5$ (a quarter of all nodes).

Table~\ref{skimruntime:table} lists selected single thread running times for reverse-rank \skim.
Listed times do not include ADS computation (see Table
~\ref{ADStable}), but this preprocessing time was only a fraction of
the total.  We note that the running time did not significantly
depend on ADS size (the parameter $k$).  The parameter $k$ can impact
running time only because it can generate longer neighborhood estimate lists $L(i)$.  The
size of these lists, even with very large $k$,  is below the
effective diameter of the graph, which was small in our data sets. 
The listed times  in the table use 
$k=128$ for the four smaller graphs  and $k=64$
for the largest one.   They correspond to computing 
the full sequence (until all nodes are covered).  Note that the running
time can be significantly reduced if we stop when a desired coverage
or number of seeds are reached.  We can also observe that the running
time grows linearly with the sample size $K$.
An interesting observation is that for the largest graphs, the computation is faster for larger values of $T$ --  This is because \skim\ works with 
  the residual problem, and its size decreases more rapidly for higher 
  values of $T$.  This is in contrast to an exact greedy computation
  (Algorithm \ref{exactgreedyim:alg}), where the running time
  increases rapidly with $T$.  Selected running times, and the slowdown factor
  with respect to reverse-rank \skim\ (including ADS computation), are provided in
  Table~\ref{exactapproximatetime:table}.  We can see that for very
  small values of $T$ the exact computation is feasible but for larger
  values of $T$, the running time and the slowdown factor
  increase rapidly.
\begin{table}\caption{Reverse-rank \skim\ running times [Sec]\label{skimruntime:table}}
{\small
\begin{tabular}{l|rrr|rrr}
Instance & \multicolumn{3}{c}{$K=100$} & \multicolumn{3}{c}{$K=500$}
  \\
 $T$: & $10^2$ & $10^3$ & $10^4$ & $10^2$ & $10^3$ & $10^4$ \\
\hline
\instance{Slashdot} &   10.22 & 16.35 &  - &     17.37 &  32.81 & - \\
\instance{Twitter} &  95.4 & 68.6 & 53.4 &  114 & 98 & 82 \\
\instance{YouTube} &
151.3 &  138.6 &  228.9 &   289.3 &  256.6 & 346.9 \\
\instance{LiveJournal} &  6533 & 3517 &  2040    &
 8261  &  6123 & 4313   
\end{tabular}
}
\end{table}
\begin{table}\caption{Exact greedy running times and speedup factor of
    reverse-rank \skim\ (including sketch computation) \label{exactapproximatetime:table}}
\begin{center}
{\small
\begin{tabular}{ll|rl}
Instance & $T$ & [hours] & $\approx$ factor \\
\hline
\instance{YouTube} & $10$ &  1.50  & $\times$5 \\
\instance{YouTube} & $10^2$ &  7.26  & $\times$23 \\
\instance{YouTube} & $10^3$  &  20.74  & $\times$68\\
\instance{LiveJournal} & $10$  &  3.01  & $\times$1 \\
\instance{LiveJournal} & $10^2$  &  11.59  & $\times$4 
\end{tabular}
}
\end{center}
\end{table}

 Finally, we  evaluate the quality of our approximate greedy sequences
  which were optimized for $\widehat{\Inf}$, in terms of the exact
  influence objective $\Inf$.  
To do so, we used a variation of Algorithm \ref{exactgreedyim:alg} to compute
the exact influence of the sequence of seeds returned by reverse-rank  \skim.
We observed that even for ADS parameter $k=64$ and $k=128$, the $\Inf$
coverage of the approximate   greedy sequence for $\widehat{\Inf}$ was
consistently within 5\% of  the exact greedy sequence for $\Inf$, and
typically much closer.

\ignore{
{\bf TO BE DONE Semantics:}
\begin{itemize}
\item 
To what extent we have asymmetry in the network?

Show asymmetry of $\pi_{ij}$ and $\pi_{ji}$. 
In undirected networks:  show distribution of ratio of 
$\max\{\pi_{ij}/\pi_{ji}, \pi{ji},\pi{ij}\}$ for random nodes with 
representative degrees.  We need to sample few hundred nodes and 
perform full Dijkstra from each.  Note that random nodes are typically 
far (so $\pi_{ij}$ in expectation is $n/2$) 
do same for closer nodes, say within 2-hops.

Show \rN\ distributions for typical core and periphery nodes. 
For core choose a higher degree node.  For periphery, a low degree 
node. 

\item 
  Use data that has communities.  Take some top communities. 

  From each community choose ten random beacons, pretend they are 
  the only labeled members of the community and we are predicting the rest. 

  Look at how other nodes rank the beacons (reverse ranks from each 
  beacon, and aggregate these ranks). 

   Use min, average, harmonic mean aggregations.  Use distance, rank,
   reverse rank.  Do ROC curves. 

  Which is the strongest signal?   perhaps need to normalize to 
  number of community memberships of a node. 
\item 
??  Compare approximation quality of \rD\ rank oracle based on 
estimated upper bound on distance. 
\end{itemize}
}

\section{Conclusion}

  Rank-based measures were used for decades as an alternative to
  distance-based measures.   Here, we defined and motivated rank-based measures of
  centrality and influence, which we believe will become important
  tools in network mining and analysis.  We then presented novel highly scalable algorithms for
  fundamental rank-based computations on graphs, including a
Dijkstra-like approximate reverse rank  single-source algorithm which
faciliates reverse-rank influence computation and reverse-rank greedy
influence maximization. 
We complement our work with hardness results that indicate that
exact computation inherently scales poorly.

 A contribution we make that is of even broader interest is a novel multithreaded design for computing
all-distance sketches (ADS) which provided the fastest implementation
for computing these sketches on multi-core architectures. This design
is relevant to many other applications of distance sketches.


  Going forward,  we plan to extend our reverse-rank IM computation to
  general decay functions, design a multithreaded implementation,
and open source our implementations.  We also hope to use our newly
available tools to explore and understand the
relation between distance-based and rank-based influence.

\ignore{
 An interesting open question is efficiently sketching the set of
 (approximate) reverse ranks, in the sense of all-distances sketches
 \cite{ECohen6f,ECohenADS:TKDE2015}.  The sketches we seek have the form
 \eqref{botkADS} with reverse rank replacing the distance.  Such
 sketches would provide influence oracles, similar to
 \cite{DSGZ:nips2013,timedinfluence:2014} for distance-based influence
 and closeness similarity oracles, similar to \cite{CDFGGW:COSN2013}
 for distance and forward-rank based similarity.
All-distances sketches can be efficiently computed with respect to
shortest-paths distances and approximate ``forward'' ranks, elapsed time and most
recent time in streams, and even distances in a set of graphs
\cite{timedinfluence:2014}.  The key to efficiency in these cases is
efficient pruning of ``reverse'' searches, which does not extend to
the reverse rank case.
}


 {\small
\notinproc{
\section*{Acknowledgements}  The authors would like to thank Shiri
Chechik and Amos Fiat for discussions and pointers.
}

\bibliographystyle{plain}
\bibliography{cycle} 
 }

\newpage
\appendix 

\begin{table}[h]
 {\small 
\begin{center}
\caption{Computing $T$ nearest neighbors for {\em all} nodes \label{KNN:table}}
\begin{tabular}{lrrr}
 & $T=16$  & $T=64$  & $T=128$  \\ 
instance &  [Sec] &  [Sec] &  [Sec] \\
\hline 
\instance{Facebook} &  0.29 & 0.65 & 0.68 \\
 \instance{Slashdot} & 9.08 & 24.45 & 45.32 \\
\instance{Twitter} & 96.78 & 222.48 & 412.54 \\
\instance{YouTube} & 1,034.93 & 1,570.93 & 2,836.12 \\
\instance{LiveJournal} & 10,589.59 & 12,885.12 & 19,183.58\\
\hline 
\end{tabular}
\end{center}
 }
\end{table}

\end{document}